\documentclass{amsart}
\usepackage{amsmath,amssymb,amsthm}
\DeclareMathOperator{\Ad}{Ad}
\DeclareMathOperator{\ad}{ad}
\newcommand\alphabar{{\bar\alpha}}
\newcommand\Bbar{{\bar B}}
\newcommand\calB{{\mathcal B}}
\newcommand\calBbar{\bar{\mathcal B}}
\newcommand\calF{{\mathcal F}}
\newcommand\calL{{\mathcal L}}
\newcommand\calLbar{\bar{\mathcal L}}
\newcommand\calM{{\mathcal M}}
\newcommand\calMbar{\bar{\mathcal M}}
\newcommand\calP{{\mathcal P}}
\newcommand\calPbar{\bar{\mathcal P}}
\newcommand\calQ{{\mathcal Q}}
\newcommand\calQbar{\bar{\mathcal Q}}
\newcommand\calX{{\mathcal X}}
\newcommand\calXbar{\bar{\mathcal X}}
\newcommand\chibar{{\bar \chi}}
\newcommand\Dbar{{\bar D}}
\newcommand\defeq{\; {\overset{\text{def}}=} \;}

\newcommand\commentout[1]{}
\newcommand\der{\partial}
\newcommand\ders{\partial_s}

\newcommand\fbar{{\bar f}}
\newcommand\gbar{{\bar g}}
\newcommand\Integer{{\mathbb Z}}
\newcommand\Lbar{{\bar L}}
\newcommand\Mbar{{\bar M}}

\DeclareMathOperator{\ordh}{ord^{\hbar}}
\DeclareMathOperator{\ordx}{ord_\xi}
\newcommand\Psibar{{\bar \Psi}}
\newcommand\Pbar{{\bar P}}
\newcommand\Qbar{{\bar Q}}
\DeclareMathOperator{\Res}{Res}
\newcommand\Sbar{{\bar S}}
\newcommand\symh{\sigma^\hbar}
\newcommand\tbar{{\bar t}}
\newcommand\totsym{\sigma_{\mathrm{tot}}}
\newcommand\ubar{{\bar u}}
\newcommand\utilde{{\tilde u}}

\newcommand\vbar{{\bar v}}
\newcommand\Wbar{{\bar W}}

\newcommand\Xbar{{\bar X}}

\newcommand\zbar{{\bar z}}
\newtheorem{thm}{Theorem}[section]
\newtheorem{lem}[thm]{Lemma}
\newtheorem{prop}[thm]{Proposition}

\theoremstyle{definition}

\theoremstyle{remark}

\newtheorem{rem}[thm]{Remark}

\numberwithin{equation}{section}

\newcommand\thmref[1]{Theorem~\ref{#1}}
\newcommand\propref[1]{Proposition~\ref{#1}}
\newcommand\secref[1]{Section~\ref{#1}}

\newcommand\lemref[1]{Lemma~\ref{#1}}

\begin{document}
\title[$\hbar$-expansion of Toda]{
       An $\hbar$-expansion of the Toda hierarchy:\\
       a recursive construction of solutions
}
\author{Kanehisa Takasaki}
\address{
Graduate School of Human and Environmental Studies,
Kyoto University,
Yoshida, Sakyo, Kyoto, 606-8501, Japan
}
\email{takasaki@math.h.kyoto-u.ac.jp}
\author{Takashi Takebe}
\address{
Faculty of Mathematics,
National Research University -- Higher School of Economics,
Vavilova Street, 7, Moscow, 117312, Russia
}
\email{ttakebe@hse.ru}
\date{2 December 2011}
\maketitle
\begin{abstract}
A construction of general solutions of the $\hbar$-dependent 
Toda hierarchy is presented.  The construction is based on 
a Riemann-Hilbert problem for the pairs $(L,M)$ and 
$(\Lbar,\Mbar)$ of Lax and Orlov-Schulman operators.  
This Riemann-Hilbert problem is translated 
to the language of the dressing operators $W$ and $\Wbar$.  
The dressing operators are set in an exponential form 
as $W = e^{X/\hbar}$ and $\Wbar = e^{\phi/\hbar}e^{\Xbar/\hbar}$, 
and the auxiliary operators $X,\Xbar$ and the function $\phi$ 
are assumed to have $\hbar$-expansions $X = X_0 + \hbar X_1 + \cdots$, 
$\Xbar = \Xbar_0 + \hbar\Xbar_1 + \cdots$ and 
$\phi = \phi_0 + \hbar\phi_1 + \cdots$.  The coefficients 
of these expansions turn out to satisfy a set of recursion 
relations.  $X,\Xbar$ and $\phi$ are recursively determined 
by these relations.  Moreover, the associated wave functions  
are shown to have the WKB form $\Psi = e^{S/\hbar}$ 
and $\Psibar = e^{\Sbar/\hbar}$, which leads to an $\hbar$-expansion 
of the logarithm of the tau function.  
\end{abstract}

\setcounter{section}{-1}
\section{Introduction}
\label{sec:intro}

This paper is a continuation of our previous work 
\cite{tak-tak:09,tak-tak:11} on a quasi-classical 
or $\hbar$-dependent (where $\hbar$ is the Planck constant) 
formulation of the KP hierarchy \cite{tak-tak:95}.   
We presented therein a recursive construction 
of general solutions to the $\hbar$-dependent KP hierarchy.  
The construction starts from a Riemann-Hilbert problem 
for the pair $(L,M)$ of Lax and Orlov-Schulman operators.  
This Riemann-Hilbert problem can be translated 
to the language of the underlying dressing operator $W$.  
Assuming the exponential form $W = e^{\hbar^{-1}X}$ 
and an $\hbar$-expansion of the operator $X$, 
one can derive a set of recursion relations 
that determine the operator $X$ order-by-order 
of the $\hbar$-expansion from the lowest part 
(namely, a solution of the dispersionless KP hierarchy 
\cite{tak-tak:95}).  Thus the Lax, Orlov-Schulman 
and dressing operators are obtained.  Furthermore, 
borrowing an idea from Aoki's exponential calculus 
of microdifferential operators \cite{aok:86}, 
one can show that the wave function has the WKB form 
$\Psi = e^{\hbar^{-1}S}$.  This leads to an $\hbar$-expansion 
of the associated tau function as a generalisation 
of the ``genus expansion'' of partition functions 
in string theories and random matrices 
\cite{dij:91,kri:91,mor:94,dfgz}.  
The goal of this paper is to generalise these results 
to an $\hbar$-dependent formulation of the Toda hierarchy 
\cite{tak-tak:95}.   

The Toda hierarchy is built from difference operators 
\begin{equation*}
  a(s,e^{\der_s}) = \sum_{m}a_m(s)e^{m\der_s} 
\end{equation*}
on a one-dimensional lattice (with coordinate $s \in \Integer$) 
rather than microdifferential operators on a continuous line. 
Even in the $\hbar$-independent case \cite{uen-tak:84}, 
the formulation of the hierarchy itself is 
more complicated than that of the KP hierarchy.  
The hierarchy has two sets of time evolutions 
for time variables $t = (t_n)$ and $\tbar = (\tbar_n)$. 
These time evolutions are formulated with two Lax operators 
$L$ and $\Lbar$.  Orlov-Schulman operators, dressing operators 
and wave functions, too, are prepared in pairs.  
In the $\hbar$-dependent formulation \cite{tak-tak:95}, 
the Planck constant $\hbar$ plays the role of lattice spacing, 
which shows up in the shift operators as $e^{\hbar\der_s}$.  
Difference operators in the Lax formalism are linear combinations 
\begin{equation*}
  a(\hbar,s,e^{\hbar\der_s}) = \sum_{m}a_m(\hbar,s)e^{m\hbar\der_s} 
\end{equation*}
of these shift operators with $\hbar$-dependent coefficients 
$a_m(\hbar,s)$.  

To construct a general solution of the $\hbar$-dependent 
Toda hierarchy, we start from a Riemann-Hilbert problem 
for the pairs $(L,M)$ and $(\Lbar,\Mbar)$ of 
Lax and Orlov-Schulman operators.  This problem can be converted 
to a problem for the dressing operators $W$ and $\Wbar$.  
We seek $W$ and $\Wbar$ in the exponential form 
\begin{equation*}
  W = e^{\hbar^{-1}X}, \quad 
  \Wbar = e^{\hbar^{-1}\phi}e^{\hbar^{-1}\Xbar}, 
\end{equation*}
where $X$ and $\Xbar$ are difference operators 
and $\phi$ is a function of $(\hbar,s,t,\tbar)$. 
Assuming that these operators and functions 
have $\hbar$-expansions, we can derive a set 
of recursion relations for the coefficients 
of these expansions.  The lowest part of this expansion 
turns out to be the dressing function 
of the dispersionless Toda hierarchy \cite{tak-tak:91,tak-tak:95}. 
Thus the description of the Lax, Orlov-Schulman 
and dressing operators are mostly parallel 
to the case of the $\hbar$-dependent KP hierarchy.  

The construction of the associated wave functions 
exhibits a new feature.  To formulate an analogue 
of Aoki's exponential calculus for difference operators, 
we define the ``symbol'' of a difference operator 
$a(\hbar,s,e^{\hbar\der_s})$ to be $a(\hbar,s,\xi)$.  
The operator product $a(\hbar,s,e^{\hbar\der_s})b(\hbar,s,e^{\hbar\der_s})$ 
induces the $\circ$-product 
\begin{equation*}
 \begin{split}
    a(\hbar,s,\xi) \circ b(\hbar,s,\xi)
    &=
    e^{\hbar\, \xi\der_\xi \der_{s'}} 
    a(\hbar,s,\xi) b(\hbar,s',\xi')|_{s'=s,\xi'=\xi}
\\
    &=
    \sum_{n=0}^\infty\frac{\hbar^n}{n!}
    (\xi\der_\xi)^n a(\hbar,s,\xi) 
        \der_{s'}^n b(\hbar,s',\xi')|_{s'=s,\xi'=\xi} 
 \end{split}
\end{equation*}
for those symbols.  Although looking very similar, 
this product structure is slightly different 
from the $\circ$-product of the symbols $a(\hbar,x,\xi)$ 
of $\hbar$-dependent microdifferential operators 
$a(\hbar,x,\hbar\der_x)$ in that $\der_\xi$ is now replaced 
with $\xi\der_\xi$.   This tiny difference, however, 
has a considerable effect; unlike $\der_\xi$, 
$\xi\der_\xi$ does not lower the order with respect to $\xi$. 
Because of this, we are forced to modify our previous method 
\cite{tak-tak:09}.

We admit that our construction of solutions 
is extremely complicated.  The recursive procedure 
is illustrated in Appendix for a special case 
that is related to $c = 1$ string theory at self-dual radius 
\cite{dmp:93,egu-kan:94,hop:94}.  
Even in this relatively simple case, we have been unable 
to derive an explicit form of the solution unless 
a half of the full time variables are set to zero.  
This is a price to pay for treating {\em general} solutions.  
In this sense, our method cannot be directly compared 
with the method in random matrix theory \cite{mor:94,dfgz}, 
in particular, Eynard and Orantin's 
topological recursion relations \cite{eyn-ora:07}.  
Their recursion relations stem from the ``loop equations'' 
for random matrices, which are constraints 
to single out a class of special solutions 
of an underlying integrable hierarchy,
while our method does not use any extra structure other than the
integrable hierarchy itself.

This paper is organised as follows.  
Section 1 is a review of the $\hbar$-dependent formulation 
of the Toda hierarchy.  The Riemann-Hilbert problem 
is also formulated therein.   Section 2 presents 
the recursive solution of the Riemann-Hilbert problem.  
The method is a rather straightforward generalisation 
of the case of the $\hbar$-dependent KP hierarchy.  
Section 3 deals with the $\hbar$-expansion of the wave function. 
Aoki's exponential calculus is reformulated for difference operators. 
Relevant recursion relations are thereby derived, 
and shown to have a solution.  
Section 4 mentions the $\hbar$-expansion of the tau function. 

\bigskip

\paragraph*{\em Acknowledgements}

The authors are grateful to Professor Akihiro Tsuchiya for drawing our
attention to this subject.

This work is partly supported by Grant-in-Aid for Scientific Research
No.\ 19540179 and No.\ 22540186 from the Japan Society for the Promotion
of Science and by the Bilateral Joint Project ``Integrable Systems,
Random Matrices, Algebraic Geometry and Geometric Invariants''
(2010--2011) of the Japan Society for the Promotion of Science and the
Russian Foundation for Basic Research. 

TT is partly supported by the grant of the National Research University
-- Higher School of Economics, Russia, for the Individual Research
Project 09-01-0047 (2009) and 10-01-0043 (2010). Part of the work was
done during his stay in the Institut Mittag-Leffler (Djursholm, Sweden)
by the programme ``Complex Analysis and Integrable Systems'' in 2011. He
thanks the Institut Mittag-Leffler and the organisers for hospitality.

\section{$\hbar$-dependent Toda hierarchy: review}
\label{sec:dtoda}

In this section we recall several facts on the Toda hierarchy depending
on a formal parameter $\hbar$ in \cite{tak-tak:95}, \S2.7. Throughout
this paper all functions are formal power series.

The $\hbar$-dependent Toda hierarchy is defined by the Lax
representation
\begin{equation}
 \begin{aligned}
     &\hbar \frac{\der L}{\der t_n}     = [ B_n,     L ],\qquad&
     &\hbar \frac{\der L}{\der \tbar_n} = [ \Bbar_n, L ],
\\
     &\hbar \frac{\der \Lbar}{\der t_n}     = [ B_n,     \Lbar ],\qquad&
     &\hbar \frac{\der \Lbar}{\der \tbar_n} = [ \Bbar_n, \Lbar ],
\\
     &B_n = (L^n)_{\geq 0}, &
     &\Bbar_n = (\Lbar^{-n})_{\le -1}, \qquad n=1,2,\dotsc,
 \end{aligned}
\label{todah}
\end{equation}
where the {\em Lax operators} $L$, $\Lbar$ are difference operators with
respect to the discrete independent variable $s\in\hbar\Integer$ of the
form
\begin{align}
    L &= e^{\hbar \ders}
      + \sum_{n=0}^\infty u_{n+1}(\hbar,s,t,\tbar)e^{-n\hbar\ders},
\label{L}
\\
    \Lbar^{-1} &= \ubar_0(\hbar,t,\tbar,s) e^{-\hbar\ders} +
    \sum_{n=0}^\infty \ubar_{n+1}(\hbar,t,\tbar,s) e^{n\hbar\ders}
\label{Lbar}
\end{align}
and $(\quad)_{\ge 0}$ and $(\quad)_{\le -1}$ are projections onto a
linear combination of $e^{n\hbar\der/\der s}$ with $n\geqq 0$ and $\leqq
-1$, respectively. Note that $e^{\hbar\ders}$ is a difference operator
with step $\hbar$: $e^{n\hbar\ders}f(s)=f(s+n\hbar)$. The coefficients
$u_n(\hbar,t,\tbar,s)$, $\ubar_n(\hbar,t,\tbar,s)$ of $L$, $\Lbar$ are
assumed to be formally regular with respect to $\hbar$:
$u_n(\hbar,t,\tbar,s) = \sum_{m=0}^\infty \hbar^m u_n^{(m)}(t,\tbar,s)$,
$\ubar_n(\hbar,t,\tbar,s) = \sum_{m=0}^\infty \hbar^m
\ubar_n^{(m)}(t,\tbar,s)$ as $\hbar \to 0$.

We define the {\em $\hbar$-order} of the difference
operator by 
\begin{equation}
    \ordh \left( \sum a_{n,m}(t,\tbar,s) \hbar^n e^{m \hbar\ders} \right)
    \defeq
    \max \left\{ -n \,\left|\, 
                 \sum_{m} a_{n,m}(t,\tbar,s) e^{m \hbar\ders} \neq 0 
    \right.\right\}.
\label{def:ordh}
\end{equation}
In particular, $\ordh\hbar=-1$, $\ordh e^{\hbar\ders}=0$. For
example, the condition which we imposed on the coefficients
$u_n(\hbar,t,\tbar,s)$ and $u_n(\hbar,t,\tbar,s)$ can be restated as
$\ordh(L) = \ordh(\Lbar) = 0$. 

The {\em principal symbol} (resp.\ the {\em symbol of order $l$})
of a difference operator $A = \sum a_{n,m}(t,\tbar,s)\hbar^n e^{m\hbar\ders}$
with respect to the $\hbar$-order is
\begin{align}
    \symh(A)
    &\defeq \sum_{m} a_{-\ordh(A),m}(t,\tbar,s) \xi^m 
\label{def:pr-symbol}
    \\
    (\text{resp. } 
    \symh_l(A)
    &\defeq \sum_{m} a_{-l,m}(t,\tbar,s) \xi^m.
\label{def:symbol}
\end{align}
When it is clear from the context, we sometimes use $\symh$ instead
of $\symh_l$.

The Lax operators $L$ and $\Lbar$ are expressed by {\em dressing
operators} $W$ and $\Wbar$:
\begin{equation}
    L     = \Ad W     (e^{\hbar\ders}) 
          = W e^{\hbar\ders} W^{-1}, \qquad
    \Lbar = \Ad \Wbar (e^{\hbar\ders}) 
          = \Wbar e^{\hbar\ders} \Wbar^{-1},
\label{L,Lbar=Ad(W,Wbar)expd}
\end{equation}
The operators $W$ and $\Wbar$ should have specific forms:
\begin{gather}
     W    = e^{\hbar^{-1} X^\circ(\hbar,t,\tbar,s,e^{\hbar\ders}) }
            e^{\hbar^{-1}\alpha(\hbar) (\hbar\ders)}
\label{W=exp(Xcirc)exp(alpha)}
\\
     X^\circ(\hbar,t,\tbar,s,e^{\hbar\ders}) =
    \sum_{k=1}^\infty \chi^\circ_k   (\hbar,t,\tbar,s) e^{-k\hbar\ders},
\label{Xcirc}
\\
    \Wbar = e^{\hbar^{-1} \phi(\hbar,t,\tbar,s)}
            e^{\hbar^{-1} \Xbar^\circ(\hbar,t,\tbar,s,e^{\hbar\ders})}
            e^{\hbar^{-1}\alphabar(\hbar) (\hbar\ders)}
\label{Wbar=exp(phi)exp(Xbarcirc)exp(alpha)}
\\
    \Xbar^\circ(\hbar,t,\tbar,s,e^{\hbar\ders}) =
    \sum_{k=1}^\infty \chibar^\circ_k(\hbar,t,\tbar,s) e^{ k\hbar\ders},
\label{Xbarcirc}
\\
\begin{aligned}
    \ordh (\phi(\hbar,t,\tbar,s)) &=
    \ordh (X^\circ(\hbar,t,\tbar,s,e^{\hbar\ders})) = 
    \ordh \alpha(\hbar)
\\
    &= \ordh (\Xbar^\circ(\hbar,t,\tbar,s,e^{\hbar\ders})) = 
    \ordh \alphabar(\hbar) = 0,
\end{aligned}
\label{ordXcircetc=0}
\end{gather}
and $\alpha(\hbar)$ and $\alphabar(\hbar)$ are constants with respect
to $t$, $\tbar$ and $s$. (In \cite{tak-tak:95} we did not introduce
$\alpha$, $\alphabar$, which will be necessary in
\secref{sec:recursion}.) 

Note that the set of operators of the form
\begin{equation}
    a\hbar\frac{\der}{\der s} + \sum_{k} \chi_k(s) e^{k\hbar\ders},
\label{diff+diff-op}
\end{equation}
where $a$ does not depend on $s$ and $\chi_k(s)$ are functions of $s$,
is closed under the commutator bracket. Hence any theorem or formula for
Lie algebras can be applied to such operators. In particular, using the
Campbell-Hausdorff formula, we can rewrite $W$ and $\Wbar$ in
the following form, which will be more convenient in the later
discussion:
\begin{gather}
    W     = e^{\hbar^{-1} X(\hbar,t,\tbar,s,e^{\hbar\ders})}
\label{W=exp(X)}
\\
    X(\hbar,t,\tbar,s,e^{\hbar\ders})
    =
    \alpha(\hbar) \hbar\frac{\der}{\der s}
    +
    \sum_{k=1}^\infty \chi_k   (\hbar,t,\tbar,s) e^{-k\hbar\ders},
\label{X}
\\
    \Wbar = e^{\hbar^{-1} \phi(\hbar,t,\tbar,s)}
            e^{\hbar^{-1} \Xbar(\hbar,t,\tbar,s,e^{\hbar\ders})},
\label{Wbar=exp(phi)exp(Xbar)}
\\
    \Xbar(\hbar,t,\tbar,s,e^{\hbar\ders})
    =
    \alphabar(\hbar) \hbar\frac{\der}{\der s}
    +
    \sum_{k=1}^\infty \chibar_k(\hbar,t,\tbar,s) e^{ k\hbar\ders},
\label{Xbar}
\\
\begin{aligned}
    \ordh (\phi(\hbar,t,\tbar,s)) &=
    \ordh (X(\hbar,t,\tbar,s,e^{\hbar\ders})) = 
    \ordh \alpha(\hbar)
\\
    &= \ordh (\Xbar(\hbar,t,\tbar,s,e^{\hbar\ders})) = 
    \ordh \alphabar(\hbar) = 0.
\end{aligned}
\label{ordXetc=0}
\end{gather}
Here we define the $\hbar$-order and the principal symbol of operators
of the form \eqref{diff+diff-op}, in particular those of $X$ and
$\Xbar$, by defining $\ordh(\hbar\ders)=0$ and
$\symh(\hbar\ders)=\log\xi$, which are consistent with the former
definitions \eqref{def:ordh} and \eqref{def:pr-symbol}.

The {\em wave functions} $\Psi(\hbar,t,\tbar,s;z)$ and
$\Psibar(\hbar,t,\tbar,s;\zbar)$ are defined by
\begin{equation}
    \Psi(\hbar,t,\tbar,s;z) 
    = W z^{s/\hbar} e^{\zeta(t,z)/\hbar},\qquad
    \Psibar(\hbar,t,\tbar,s;\zbar) 
    = \Wbar \zbar^{s/\hbar} e^{\zeta(\tbar,\zbar^{-1})/\hbar},
\label{def:wave-func}
\end{equation}
where $\zeta(t,z)=\sum_{n=1}^\infty t_n z^n$,
$\zeta(\tbar,\zbar^{-1})=\sum_{n=1}^\infty \tbar_n \zbar^{-n}$. They
are solutions of linear equations
\begin{align*}
    L \Psi &= z \Psi, &
    \hbar\frac{\der \Psi}{\der t_n} &= B_n \Psi,&
    \hbar\frac{\der \Psi}{\der \tbar_n} &= \Bbar_n \Psi,
    \quad (n=1,2,\dotsc),
\\
    \Lbar \Psibar &= \zbar \Psibar, &
    \hbar\frac{\der \Psibar}{\der t_n} &= B_n \Psibar,&
    \hbar\frac{\der \Psibar}{\der \tbar_n} &= \Bbar_n \Psibar,
    \quad (n=1,2,\dotsc),
\end{align*}
and have the WKB form \eqref{wave-func}, as we shall show in
\secref{sec:wave-function}. Moreover they are expressed by means of the
{\em tau function} $\tau(\hbar,t,\tbar,s)$ as follows:
\begin{equation}
 \begin{aligned}
    \Psi(\hbar,t,\tbar;z) &=
    \frac{ \tau(\hbar,t-\hbar [z^{-1}] ,\tbar,s) }
         { \tau(\hbar,t,\tbar,s) }
    z^{\alpha(\hbar)/\hbar} 
    z^{s/\hbar} e^{\zeta(t,z)/\hbar},
\\
    \Psibar(\hbar,t,\tbar;\zbar) &=
    \frac{ \tau(\hbar,t,\tbar-\hbar [\zbar] ,s+\hbar) }
         { \tau(\hbar,t,\tbar,s) }
    \zbar^{\alphabar(\hbar)/\hbar} 
    \zbar^{s/\hbar} e^{\zeta(\tbar,\zbar^{-1})/\hbar}
 \end{aligned}
\label{tau/tau}
\end{equation}
where $[z^{-1}]=(1/z,1/2z^2,1/3z^3,\dots)$,
$[\zbar]=(\zbar,\zbar^2/2,\zbar^3/3,\dotso)$. We shall study the
$\hbar$-expansion of the tau function in \secref{sec:tau-function}.

The {\em Orlov-Schulman operators} $M$ and $\Mbar$ \cite{orl-sch} are
defined by
\begin{align}
    M &= 
    \Ad \left(
          W \exp\left(\hbar^{-1} \zeta(t,e^{\hbar\ders}) \right)
        \right)s &
    &=
    W \left( \sum_{n=1}^\infty n t_n e^{n\hbar\ders} + s \right)
    W^{-1}
\label{def:M}
\\
    \Mbar &= 
    \Ad \left(
          \Wbar \exp\left(\hbar^{-1} \zeta(\tbar,e^{-\hbar\ders})
                    \right)
        \right)s &
    &=
    \Wbar
     \left( -\sum_{n=1}^\infty n t_n e^{-n\hbar\ders} + s \right)
    \Wbar^{-1}
\label{def:Mbar}
\end{align}
where $\zeta(t,e^{\hbar\ders}) = \sum_{n=1}^\infty t_n e^{n\hbar\ders}$
and $\zeta(\tbar,e^{-\hbar\ders}) = \sum_{n=1}^\infty \tbar_n
e^{-n\hbar\ders}$. It is easy to see that $M$ and $\Mbar$ have forms
\begin{align}
    M &= \sum_{n=1}^\infty n t_n L^n + s + \alpha(\hbar)
       + \sum_{n=1}^\infty v_n(\hbar,t,\tbar,s) L^{-n},
\label{M}
\\
    \Mbar &=-\sum_{n=1}^\infty n \tbar_n \Lbar^{-n} + s
           + \alphabar(\hbar)
           + \sum_{n=1}^\infty \vbar_n(\hbar,t,\tbar,s) \Lbar^n.
\label{Mbar}
\end{align}
and satisfies the following properties:
\begin{itemize}
\item $\ordh(M) = \ordh(\Mbar) = 0$;
\item the canonical commutation relation: $[L, M] = \hbar L$ and
      $[\Lbar, \Mbar] = \hbar \Lbar$;
\item the same Lax equations as $L$, $\Lbar$:
\begin{equation}
 \begin{aligned}
    \hbar \frac{\der M}{\der t_n}     &= [ B_n, M ], \quad&
    \hbar \frac{\der M}{\der \tbar_n} &= [\Bbar_n, M ],
\\
    \hbar \frac{\der \Mbar}{\der t_n}     &= [ B_n, \Mbar ], \quad&
    \hbar \frac{\der \Mbar}{\der \tbar_n} &= [ \Bbar_n, \Mbar ],
 \end{aligned}
    \quad n = 1,2,\ldots;
\label{lax:M,Mbar}
\end{equation}
\item another linear equation for the wave function $\Psi$:
\begin{equation*}
    M \Psi = \hbar z \frac{\der \Psi}{\der z}, \qquad
    \Mbar \Psibar = \hbar \zbar \frac{\der\Psibar}{\der \zbar}.
\end{equation*}
\end{itemize}

\begin{rem}
\label{rem:alpha=const}
 As in the KP case (Remark 1.2 in \cite{tak-tak:09}), if an operator $M$
 of the form \eqref{M} and and an operator $\Mbar$ of the form
 \eqref{Mbar} satisfy the Lax equations \eqref{lax:M,Mbar} and the
 canonical commutation relation $[L,M]=\hbar L$ and
 $[\Lbar,\Mbar]=\hbar\Lbar$ with the Lax operator $L$ and $\Lbar$ of the
 Toda lattice hierarchy, then $\alpha(\hbar)$ and $\alphabar(\hbar)$ in
 the expansions \eqref{M} and \eqref{Mbar} do not depend on any $t_n$,
 $\tbar_n$ nor $s$. In fact, suppose that $\alpha(\hbar)$ depens on $s$;
 $\alpha(\hbar)=\alpha(\hbar,s)$. Then, the canonical commutation
 relation $[L,M]=\hbar L$ is expanded as
\[
    (\hbar+\alpha(\hbar,s+\hbar)-\alpha(\hbar,s)) e^{\hbar\ders}
    + (\text{difference operators of lower order}) 
    = \hbar L,
\]
 which implies $\alpha(\hbar,s+\hbar)=\alpha(\hbar,s)$. Similarly, from
 \eqref{lax:M,Mbar} follows $\frac{\der\alpha}{\der
 t_n}=\frac{\der\alpha}{\der \tbar_n}=0$ with the help of \eqref{todah}
 and $[L^n, M] = n\hbar L^n$. For $\alphabar(s)$ the proof is the same.
\end{rem}

The following proposition (Proposition 2.7.11 of \cite{tak-tak:95}) is a
``dispersionful'' counterpart of the theorem for the dispersionless Toda
hierarchy found earlier (cf.\ Section 4 of \cite{tak-tak:91} and 
\propref{prop:dRH} below).

\begin{prop}
\label{prop:RH}
(i)
 Suppose that operators 
 $f(\hbar,s,e^{\hbar\ders})$, $g(\hbar,s,e^{\hbar\ders})$,
 $\fbar(\hbar,s,e^{\hbar\ders})$, $\gbar(\hbar,s,e^{\hbar\ders})$,
 $L$, $\Lbar$, $M$ and $\Mbar$ satisfy the following conditions:
\begin{itemize}
 \item $\ordh f = \ordh g = \ordh \fbar = \ordh \gbar = 0$,
       $[f,g]=\hbar f$, $[\fbar,\gbar]=\hbar\fbar$;
 \item $L$, $\Lbar$, $M$ and $\Mbar$ are of the form \eqref{L},
       \eqref{Lbar}, \eqref{M} and \eqref{Mbar} respectively. They are
       canonically commuting: $[L,M]=\hbar L$,
       $[\Lbar,\Mbar]=\hbar\Lbar$; 
 \item Equations
\begin{equation}
    f(\hbar,M,L) = \fbar(\hbar,\Mbar,\Lbar), \qquad
    g(\hbar,M,L) = \gbar(\hbar,\Mbar,\Lbar)
\label{(f,g)(M,L)=(fbar,gbar)(Mbar,Lbar)}
\end{equation}
       hold.
\end{itemize}
 Then the pair $(L,\Lbar)$ is a solution of the Toda lattice hierarchy
 \eqref{todah} and $M$ and $\Mbar$ are the corresponding Orlov-Schulman
 operators.

(ii)
 Conversely, for any solution $(L,\Lbar,M,\Mbar)$ of the
 $\hbar$-dependent Toda lattice hierarchy there exists a quadruplet
 $(f,\fbar,g,\gbar)$ satisfying the conditions in (i).
\end{prop}

The leading term of the $\hbar$-dependent Toda lattice hierarchy with
respect to the $\hbar$-order gives the {\em dispersionless Toda
hierarchy}. Namely,
\begin{align}
    \calL &:= \symh(L)
    = \xi
    + \sum_{n=0}^\infty u_{0,n+1} \xi^{-n}, \qquad
    (u_{0,n+1} := \symh(u_{n+1})),
\label{calL}
\\
    \calLbar^{-1} &= \symh(\Lbar^{-1})
    = \ubar_{0,0} \xi^{-1} +
    \sum_{n=0}^\infty \ubar_{0,n+1} \xi^n, \qquad
    (\ubar_{0,n+1} := \symh(\ubar_{n+1}))
\label{calLbar}
\end{align}
satisfy the dispersionless Lax type equations
\begin{equation}
 \begin{aligned}
     &\frac{\der \calL}{\der t_n}     = \{\calB_n,    \calL\},\qquad&
     &\frac{\der \calL}{\der \tbar_n} = \{\calBbar_n, \calL\},
\\
     &\frac{\der \calLbar}{\der t_n}     = \{\calB_n,    \calLbar\},
     \qquad&
     &\frac{\der \calLbar}{\der \tbar_n} = \{\calBbar_n, \calLbar\},
\\
     &\calB_n = (\calL^n)_{\geq 0}, &
     &\calBbar_n = (\calLbar^{-n})_{\le -1}, \qquad n=1,2,\dotsc,
 \end{aligned}
\label{dtoda}
\end{equation}
where $(\quad)_{\geq 0}$ and $(\quad)_{\geq 0}$ are the truncation of
Laurent series to the polynomial part and to the negative order part
respectively. The Poisson bracket $\{,\}$ is defined by
\begin{equation}
    \{a(s,\xi), b(s,\xi)\}
    =
    \xi\left(
    \frac{\der a}{\der \xi} \frac{\der b}{\der s}
    -
    \frac{\der a}{\der s} \frac{\der b}{\der \xi}
    \right).
\label{poisson}
\end{equation}

The dressing operation \eqref{L,Lbar=Ad(W,Wbar)expd} for $L$ and $\Lbar$
becomes the following dressing operation for $\calL$ and $\calLbar$:
\begin{equation}
 \begin{aligned}
    \calL
    &= \exp \bigl( \ad_{\{,\}} X_0 \bigr) \xi,
    &
    X_0     &:= \symh(X),
\\
    \calLbar
    &= \exp \bigl( \ad_{\{,\}} \phi_0 \bigr) 
       \exp \bigl( \ad_{\{,\}} \Xbar_0 \bigr) \xi,
    &
    \phi_0 &:= \symh(\phi),\ 
    \Xbar_0 := \symh(\Xbar),
\end{aligned}
\label{calL,Lbar=exp(adX,Xbar)xi}
\end{equation}
where $\ad_{\{,\}} (f) (g):= \{f,g\}$. 

The principal symbol of the Orlov-Schulman operators are {\em
Orlov-Schulman functions},
\begin{align}
    \calM &= \sum_{n=1}^\infty nt_n \calL^n + s
    + \alpha_0
    + \sum_{n=1}^\infty v_{0,n} \calL^{-n},
\label{calM}
\\
    (v_{0,n}&:=\symh(v_n),\quad
    \alpha_0:=\symh(\alpha(\hbar)) )
\nonumber
\\ 
   \calMbar &= - \sum_{n=1}^\infty n\tbar_n \calLbar^{-n} + s
    + \alphabar_0 
    + \sum_{n=1}^\infty \vbar_{0,n} \calLbar^n,
\label{calMbar}
\\
    (\vbar_{0,n}&:=\symh(\vbar_n), \quad
    \alphabar_0:=\symh(\alphabar_0(\hbar))).
\nonumber
\end{align}
which are equal to
\begin{align}
    \calM &= 
    \exp \bigl( \ad_{\{,\}} X_0 \bigr) 
    \exp \bigl( \ad_{\{,\}} \zeta(t,\xi)\bigr) s,
\label{calM=exp(adX)s}
\\
    \calMbar &= 
    \exp \bigl( \ad_{\{,\}} \phi_0 \bigr)
    \exp \bigl( \ad_{\{,\}} \Xbar_0 \bigr) 
    \exp \bigl( \ad_{\{,\}} \zeta(\tbar,\xi^{-1})\bigr) s
\label{calMbar=exp(adXbar)s}
\end{align}
where $\zeta(t,\xi) = \sum_{n=1}^\infty t_n \xi^n$ and
$\zeta(\tbar,\xi^{-1}) = \sum_{n=1}^\infty \tbar_n \xi^{-n}$. The series
$\calM$ satisfies the canonical commutation relation with $\calL$,
$\{\calL,\calM\}=\calL$, while $\calMbar$ satisfies the canonical
commutation relation with $\calLbar$,
$\{\calLbar,\calMbar\}=\calLbar$. The principal symbols of equations
\eqref{lax:M,Mbar} give the Lax type equations:
\begin{equation}
 \begin{aligned}
    \hbar \frac{\der \calM}{\der t_n}     &= \{\calB_n,    \calM\},
    \quad&
    \hbar \frac{\der \calM}{\der \tbar_n} &= \{\calBbar_n, \calM\},
\\
    \hbar \frac{\der \calMbar}{\der t_n}  
    &= \{\calB_n, \calMbar\},
    \quad&
    \hbar \frac{\der \calMbar}{\der \tbar_n}
    &= \{\calBbar_n, \calMbar\},
 \end{aligned}
    \quad n = 1,2,\ldots,
\label{lax:calM,Mbar}
\end{equation}

The Riemann-Hilbert type construction of the solution is essentially the
same as \propref{prop:RH}. (Proposition 2.5.1 of \cite{tak-tak:95}; We
do not need to assume the canonical commutation relation
$\{\calL,\calM\}=\calL$ and $\{\calLbar,\calMbar\}=\calLbar$.)

\begin{prop}
\label{prop:dRH}
(i)
 Suppose that functions $f_0(s,\xi)$, $g_0(s,\xi)$, $\fbar_0(s,\xi)$,
 $\gbar_0(s,\xi)$, $\calL$, $\calLbar$, $\calM$ and $\calMbar$
 satisfy the following conditions:
\begin{itemize}
 \item $\{f_0,g_0\}=f_0$, $\{\fbar_0,\gbar_0\}=\fbar_0$;
 \item $\calL$, $\calLbar^{-1}$, $\calM$ and $\calMbar$ have the form
       \eqref{calL}, \eqref{calLbar}, \eqref{calM} and \eqref{calMbar}
       respectively. 
 \item Equations
\begin{equation}
    f_0(\calM,\calL) = \fbar_0(\calMbar,\calLbar), \qquad
    g_0(\calM,\calL) = \gbar_0(\calMbar,\calLbar).
\label{(f0,g0)(M,L)=(fbar0,gbar0)(Mbar,Lbar)}
\end{equation}
       hold.
\end{itemize}
 Then the pair $(\calL,\calLbar)$ is a solution of the dispersionless
 Toda hierarchy \eqref{dtoda} and $\calM$ and $\calMbar$ are the
 corresponding Orlov-Schulman functions.

(ii)
 Conversely, for any solution $(\calL,\calLbar,\calM,\calMbar)$ of the
 dispersionless Toda hierarchy, there exists a quadruplet
 $(f_0,g_0,\fbar_0,\gbar_0)$  satisfying the conditions in (i).
\end{prop}

If $(f,g,\fbar,\gbar)$, $(L,\Lbar,M,\Mbar)$ are as in \propref{prop:RH},
then
$(f_0=\symh(f),g_0=\symh(g),\fbar_0=\symh(\fbar),\gbar_0=\symh(\gbar))$,
$(\calL=\symh(L),\calLbar=\symh(\Lbar),\calM=\symh(M),\calMbar=\symh(\Mbar))$
satisfy the conditions in \propref{prop:dRH}. In other words,
$(f,g,\fbar,\gbar)$ and $(L,\Lbar,M,\Mbar)$ are quantisation of
$(f_0,g_0,\fbar_0,\gbar_0)$ and $(\calL,\calM,\calLbar,\calMbar)$
respectively. (See, for example, \cite{sch:85} for quantised canonical
transformations.)

\section{Recursive construction of the dressing operator}
\label{sec:recursion}

In this section we prove that the solution of the $\hbar$-dependent Toda
lattice hierarchy corresponding to $(f,g,\fbar,\gbar)$ in
\propref{prop:RH} is recursively constructed from its leading term,
i.e., the solution of the dispersionless Toda hierarchy corresponding to
the Riemann-Hilbert data
$(\symh(f),\symh(g),\symh(\fbar),\symh(\gbar))$.

Given the quadruplet $(f,g,\fbar,\gbar)$, we have to construct the
dressing operator $W$ and $\Wbar$, or, in other words, $X$ in
\eqref{W=exp(X)} and $\phi$, $\Xbar$ in \eqref{Wbar=exp(phi)exp(Xbar)},
such that equations \eqref{(f,g)(M,L)=(fbar,gbar)(Mbar,Lbar)} hold, or
equivalently, the following equations hold:
\begin{equation}
 \begin{aligned}
    &\Ad \left(
         W \exp\left(\hbar^{-1} \zeta(t,e^{\hbar\ders}) \right)
         \right) f(\hbar,s,e^{\hbar\ders})
\\
    ={}
    &\Ad \left(
         \Wbar \exp\left(\hbar^{-1} \zeta(\tbar,e^{-\hbar\ders}) \right)
         \right) \fbar(\hbar,s,e^{\hbar\ders}),
\\
    &\Ad \left(
         W \exp\left(\hbar^{-1} \zeta(t,e^{\hbar\ders}) \right)
         \right) g(\hbar,s,e^{\hbar\ders})
\\  ={}
    &\Ad \left(
         \Wbar \exp\left(\hbar^{-1} \zeta(\tbar,e^{-\hbar\ders}) \right)
         \right) \gbar(\hbar,s,e^{\hbar\ders}).
 \end{aligned}
\label{Ad(W)(f,g)=Ad(Wbar)(fbar,gbar)}
\end{equation}
Let us expand $X$, $\Xbar$ and $\phi$ with respect to the $\hbar$-order
as follows:
\begin{align}
    X &= \sum_{n=0}^\infty \hbar^n X_n,\quad
    X_n = X_n(t,\tbar,s,e^{\hbar\ders}) =
    \alpha_n \hbar\frac{\der}{\der s}
    +
    \sum_{k=1}^\infty \chi_{n,k}(t,\tbar,s) e^{-k\hbar\ders},
\label{X:h-expansion}
\\
    \Xbar &= \sum_{n=0}^\infty \hbar^n \Xbar_n,\quad
    \Xbar_n = \Xbar_n(t,\tbar,s,e^{\hbar\ders}) =
    \alphabar_n \hbar\frac{\der}{\der s}
    +
    \sum_{k=1}^\infty \chibar_{n,k}(t,\tbar,s) e^{ k\hbar\ders},
\label{Xbar:h-expansion}
\\
    \phi &= \sum_{n=0}^\infty \hbar^n \phi_n(t,\tbar,s), \quad
    \phi_n = \phi_n(t,\tbar,s),
\label{phi:h-expansion}
\end{align}
where $\alpha_n$, $\alphabar_n$, $\chi_{n,k}$ and $\chibar_{n,k}$ do not
depend on $\hbar$, and hence $\alpha$, $\chi_k$ in \eqref{X}
and $\alphabar$, $\chibar_k$ in \eqref{Xbar} are expanded as
$\alpha=\sum_{n=0}^\infty\hbar^n\alpha_n$,
$\chi_k=\sum_{n=0}^\infty\hbar^n\chi_{n,k}$,
$\alphabar=\sum_{n=0}^\infty\hbar^n\alphabar_n$ and
$\chibar_k=\sum_{n=0}^\infty\hbar^n\chibar_{n,k}$.

Assume that a solution of the dispersionless Toda hierarchy
corresponding to $(\symh(f),\symh(g),\symh(\fbar),\symh(\gbar))$ is
given. In other words, assume that symbols
$X_0=\alpha_0\log\xi+\sum_{k=1}^\infty\chi_{0,k}(t,\tbar,s) \xi^{-k}$,
$\Xbar_0=\alphabar_0\log\xi+\sum_{k=1}^\infty\chibar_{0,k}(t,\tbar,s)
\xi^{k}$ and $\phi_0=\phi_0(t,\tbar,s)$ are given such that
\begin{equation*}
    \symh(f)(\calM,\calL) = \symh(\fbar)(\calMbar,\calLbar),
\end{equation*}
namely,
\begin{equation*}
 \begin{aligned}
    &\exp \bigl( \ad_{\{,\}} X_0 \bigr)
    \exp \bigl( \ad_{\{,\}} \zeta(t,\xi) \bigr)
    \symh(f)(s,\xi)
\\
    ={}&
    \exp \bigl( \ad_{\{,\}} \phi_0 \bigr)
    \exp \bigl( \ad_{\{,\}} \Xbar_0 \bigr)
    \exp \bigl( \ad_{\{,\}} \zeta(\tbar,\xi^{-1}) \bigr)
    \symh(\fbar)(s,\xi),
 \end{aligned}
\end{equation*}
and
\begin{equation*}
    \symh(g)(\calM,\calL) = \symh(\gbar)(\calMbar,\calLbar),
\end{equation*}
namely,
\begin{equation}
 \begin{aligned}
    &\exp \bigl( \ad_{\{,\}} X_0 \bigr)
    \exp \bigl( \ad_{\{,\}} \zeta(t,\xi) \bigr)
    \symh(g)(s,\xi)
\\
    ={}&
    \exp \bigl( \ad_{\{,\}} \phi_0 \bigr)
    \exp \bigl( \ad_{\{,\}} \Xbar_0 \bigr)
    \exp \bigl( \ad_{\{,\}} \zeta(\tbar,\xi^{-1}) \bigr)
    \symh(\gbar)(s,\xi).
 \end{aligned}
\label{sym(f,g)(M,L)=sym(fbar,gbar)(Mbar,Lbar)}
\end{equation}
(See \propref{prop:dRH}.)

We are to construct $X_n$, $\Xbar_n$ and $\phi_n$ recursively, starting
from $X_0$, $\Xbar_0$ and $\phi_0$. For this purpose expand both sides
of equations \eqref{Ad(W)(f,g)=Ad(Wbar)(fbar,gbar)} as follows:
\begin{align}
    P&:= 
    \Ad \left( \exp (\hbar^{-1} X) \right) f_t 
    = \sum_{k=0}^\infty \hbar^k P_k,
\label{P}
\\
    Q&:= 
    \Ad \left( \exp (\hbar^{-1} X) \right) g_t 
    = \sum_{k=0}^\infty \hbar^k Q_k,
\label{Q}
\\
    \Pbar&:= 
    \Ad \left( \exp (\hbar^{-1} \phi) \exp (\hbar^{-1} \Xbar) \right)
    \fbar_\tbar 
    = \sum_{k=0}^\infty \hbar^k \Pbar_k,
\label{Pbar}
\\
    \Qbar&:= 
    \Ad \left( \exp (\hbar^{-1} \phi) \exp (\hbar^{-1} \Xbar) \right)
    \gbar_\tbar
    = \sum_{k=0}^\infty \hbar^k \Qbar_k,
\label{Qbar}
\end{align}
where 
\begin{align}
    f_t &:=
    \Ad \left( e^{\hbar^{-1} \zeta(t,e^{\hbar\ders})} \right) f, &
    g_t &:=
    \Ad \left( e^{\hbar^{-1} \zeta(t,e^{\hbar\ders})} \right) g,
\label{ft,gt}
\\
    \fbar_\tbar &:=
    \Ad \left( e^{\hbar^{-1} \zeta(\tbar,e^{-\hbar\ders})} \right) 
    \fbar, &
    \gbar_\tbar &:=
    \Ad \left( e^{\hbar^{-1} \zeta(\tbar,e^{-\hbar\ders})} \right) 
    \gbar,
\label{fbart,gbart}
\end{align}
and $P_i$'s, $Q_i$'s, $\Pbar_i$'s and $\Qbar_i$'s are difference
operators of the $\hbar$-order $0$: 
\begin{equation*}
    \ordh P_i=\ordh Q_i=\ordh\Pbar_i=\ordh \Qbar_i=0.
\end{equation*}
Suppose that we have chosen $X_0,\dotsc,X_{i-1}$,
$\Xbar_0,\dotsc,\Xbar_{i-1}$ and $\phi_0,\dotsc,\phi_{i-1}$ so that
$P_j=\Pbar_j$ ($0\leqq j \leqq i-1$) and $Q_j=\Qbar_j$ ($0\leqq j \leqq
i-1$). If operators $X_i$, $\Xbar_i$ and a function $\phi_i$ are
constructed from these given $X_j$, $\Xbar_j$ and $\phi_j$ ($0\leqq j
\leqq i-1$) so that equations $P_i=\Pbar_i$ and $Q_i=\Qbar_i$ hold, this
procedure gives recursive construction of $X$, $\Xbar$ and $\phi$ in
question.

We can construct such $X_i$, $\Xbar_i$ and $\phi_i$ as follows. (Details
and meaning shall be explained in the proof of
\thmref{thm:recursion:X,Xbar}.): 
\begin{itemize}
 \item (Step 0) Assume $X_j$, $\Xbar_j$ and $\phi_j$ ($0\leqq j \leqq
       i-1$) are given and set
\begin{equation}
 \begin{aligned}
    X^{(i-1)} &:= \sum_{n=0}^{i-1} \hbar^n X_n, & &
\\
    \Xbar^{(i-1)} &:= \sum_{n=0}^{i-1} \hbar^n \Xbar_n, \qquad &
    \phi^{(i-1)} &:= \sum_{n=0}^{i-1} \hbar^n \phi_n.
\end{aligned}
\label{def:X,Xbar,phi(i-1)}
\end{equation}
 \item (Step 1) Set
\begin{align}
    P^{(i-1)} 
    &:=
    \Ad \left(\exp \hbar^{-1} X^{(i-1)}\right) f_t,
\label{def:Pi-1}
\\
    Q^{(i-1)}
    &:=
    \Ad \left(\exp \hbar^{-1} X^{(i-1)}\right) g_t,
\label{def:Qi-1}
\\
    \Pbar^{(i-1)} 
    &:=
    \Ad \left(\exp \hbar^{-1} \phi^{(i-1)}\right)
    \Ad \left(\exp \hbar^{-1} \Xbar^{(i-1)}\right) \fbar_\tbar,
\label{def:Pbari-1}
\\
    \Qbar^{(i-1)}
    &:=
    \Ad \left(\exp \hbar^{-1} \phi^{(i-1)}\right)
    \Ad \left(\exp \hbar^{-1} \Xbar^{(i-1)}\right) \gbar_\tbar,
\label{def:Qbari-1}
\end{align}
       and expand $P^{(i-1)}$ and $Q^{(i-1)}$ with respect to the
       $\hbar$-order as
\begin{align}
    P^{(i-1)} &= \sum_{k=0}^\infty \hbar^k P^{(i-1)}_k,&
    Q^{(i-1)} &= \sum_{k=0}^\infty \hbar^k Q^{(i-1)}_k,&
\label{Pi-1,Qi-1:h-expand}
\\
    \Pbar^{(i-1)} &= \sum_{k=0}^\infty \hbar^k \Pbar^{(i-1)}_k,&
    \Qbar^{(i-1)} &= \sum_{k=0}^\infty \hbar^k \Qbar^{(i-1)}_k.&
\label{Pbari-1,Qbari-1:h-expand}
\end{align}
        ($\ordh P^{(i-1)}_k=\ordh
       Q^{(i-1)}_k=\ordh\Pbar^{(i-1)}_k=\ordh\Qbar^{(i-1)}_k=0$.)
 \item (Step 2) Put
\begin{align*}
    \calP_0 &:=\symh(P^{(i-1)}_0), &
    \calQ_0 &:=\symh(Q^{(i-1)}_0), 
\\
    \calP^{(i-1)}_i &:=\symh(P^{(i-1)}_i), &
    \calQ^{(i-1)}_i &:=\symh(Q^{(i-1)}_i), 
\\
    \calPbar_0 &:=\symh(\Pbar^{(i-1)}_0), &
    \calQbar_0 &:=\symh(\Qbar^{(i-1)}_0), 
\\
    \calPbar^{(i-1)}_i &:=\symh(\Pbar^{(i-1)}_i), &
    \calQbar^{(i-1)}_i &:=\symh(\Qbar^{(i-1)}_i) 
\end{align*}
       and define series
$
    \tilde\calX_i(t,\tbar,s,\xi)
    =
    \alpha_i\log\xi +
    \sum_{k=1}^\infty \tilde\chi_{i,k}(t,\tbar,s)\xi^{-k}
$,  
$
    \tilde\calXbar_i(t,\tbar,s,\xi)
    =
    \alphabar_i\log\xi +
    \sum_{k=1}^\infty \tilde\chibar_{i,k}(t,\tbar,s)\xi^{-k}
$
       and a function $\phi_i(t,\tbar,s)$ by one of the following
       integrals. (The integrand of the first integral in the right
       hand side of each equation is considered as a series of $\xi$
       around $\xi=\infty$ and the integrand of the second integral
       is considered as a series around $\xi=0$.)
\begin{align}
    -\tilde\calX_i + \phi_i + \tilde\calXbar_i
    &=
    \int^\xi \calP_0^{-1} 
    \left(
    -
    \dfrac{\der \calQ_0}{\der \xi} \calP^{(i-1)}_i
    +
    \dfrac{\der \calP_0}{\der \xi} \calQ^{(i-1)}_i
    \right)
    d\xi
\label{tildeXi=intxi}
\\
    &-
    \int^\xi \calPbar_0^{-1} 
    \left(
    -
    \dfrac{\der \calQbar_0}{\der \xi} \calPbar^{(i-1)}_i
    +
    \dfrac{\der \calPbar_0}{\der \xi} \calQbar^{(i-1)}_i
    \right)
    d\xi,
\nonumber
\\
    -\tilde\calX_i + \phi_i + \tilde\calXbar_i 
    &=
    \int^s \calP_0^{-1} 
    \left(
    -
    \dfrac{\der \calQ_0}{\der s} \calP^{(i-1)}_i
    +
    \dfrac{\der \calP_0}{\der s} \calQ^{(i-1)}_i
    \right)
    ds
\label{tildeXi=ints}
\\
    &-
    \int^s \calPbar_0^{-1} 
    \left(
    -
    \dfrac{\der \calQbar_0}{\der s} \calPbar^{(i-1)}_i
    +
    \dfrac{\der \calPbar_0}{\der s} \calQbar^{(i-1)}_i
    \right)
    ds,
\nonumber
\end{align}
       In fact they give the same $\tilde\calX_i$, $\phi_i$ and
       $\tilde\calXbar_i$. Exactly speaking, the coefficients of $\xi^n$
       ($n\in\Integer$, $n\neq0$) and $\log\xi$ are determined by the
       first equation \eqref{tildeXi=intxi} and its integral constant
       $\phi_i$ is determined by the second equation
       \eqref{tildeXi=ints} up to an arbitrary additive constant.

       Equations \eqref{tildeXi=intxi} and \eqref{tildeXi=ints}
       determine the combination $-\alpha_i+\alphabar_i$ as the
       coefficient of $\log\xi$ but not $\alpha_i$ nor $\alphabar_i$
       separately, which can be chosen arbitrarily so far as
       $-\alpha_i+\alphabar_i$ is fixed.
 \item (Step 3) Define a series
$
    \calX_i(t,\tbar,s,\xi)
    =
    \alpha_i\log\xi
    +
    \sum_{k=1}^\infty \chi_{i,k}(t,\tbar,s)\xi^{-k} 
$
and 
$
    \calXbar_i(t,\tbar,s,\xi)
    =
    \alphabar_i\log\xi
    +
    \sum_{k=1}^\infty \chibar_{i,k}(t,\tbar,s)\xi^{-k} 
$
       by
\begin{equation}
 \begin{split}
    \calX_i
    &= \tilde \calX_i
    - \frac{1}{2} \{\symh(X_0),\tilde\calX_i\}
    + \sum_{p=1}^\infty K_{2p}
      (\ad_{\{,\}} (\symh(X_0)))^{2p} \tilde\calX_i,
\\
    \calXbar_i
    &= \tilde \calXbar'_i
    - \frac{1}{2} \{\symh(\Xbar_0),\tilde\calXbar'_i\}
    + \sum_{p=1}^\infty K_{2p}
      (\ad_{\{,\}} (\symh(\Xbar_0)))^{2p} \tilde\calXbar'_i,
\\
    \tilde\calXbar'_i &:= e^{-\ad_{\{,\}} \phi_0} \tilde\calXbar_i
 \end{split}
\label{tildeXi,Xbari->Xi,Xbari:symbol}
\end{equation}
       Here $K_{2p}$ is determined by the generating function
\begin{equation}
    \frac{z}{e^z-1}
    = 1 - \frac{z}{2} + \sum_{p=1}^\infty K_{2p} z^{2p},
\label{def:K2p}
\end{equation}
       i.e., $K_{2p}= B_{2p}/(2p)!$, where $B_{2p}$'s are the Bernoulli
       numbers. 
 \item (Step 4) The operators $X_i(t,\tbar,s,e^{\hbar\ders})$ and
       $\Xbar_i(t,\tbar,s,e^{\hbar\ders})$ are defined as the
       operators with the principal symbols $\calX_i$ and $\calXbar_i$:
\begin{equation}
    X_i = \sum_{k=1}^\infty\chi_{i,k}(t,\tbar,s) e^{-k\hbar\ders}, \qquad
    \Xbar_i = \sum_{k=1}^\infty\chibar_{i,k}(t,\tbar,s) e^{k\hbar\ders}.
\label{def:Xi,Xbari}
\end{equation}
\end{itemize}

The main theorem is the following:
\begin{thm}
\label{thm:recursion:X,Xbar}
 Assume that $X_0$, $\Xbar_0$ and $\phi_0$ satisfy
 \eqref{sym(f,g)(M,L)=sym(fbar,gbar)(Mbar,Lbar)} and construct $X_i$'s,
 $\Xbar_i's$ and $\phi_i$'s by the above procedure recursively. Then
 $X$, $\Xbar$ and $\phi$ defined by \eqref{X:h-expansion},
 \eqref{Xbar:h-expansion} and \eqref{phi:h-expansion} satisfy
 \eqref{(f,g)(M,L)=(fbar,gbar)(Mbar,Lbar)}. Namely
 $W=\exp(X/\hbar)$ and $\Wbar=\exp(\phi/\hbar) \exp(\Xbar/\hbar)$ are
 dressing operators of the $\hbar$-dependent Toda hierarchy
 corresponding to the data $(f,g,\fbar,\gbar)$.
%
\end{thm}

\bigskip
The rest of this section is the proof of \thmref{thm:recursion:X,Xbar}
by induction. The essential idea of the proof is almost the same as the
proof of Theorem 2.1 in \cite{tak-tak:09}.

Let us denote the ``known'' part of $X$, $\Xbar$ and $\phi$ by
$X^{(i-1)}$, $\Xbar^{(i-1)}$ and $\phi^{(i-1)}$ as in
\eqref{def:X,Xbar,phi(i-1)} and, as intermediate objects, consider
$P^{(i-1)}$, $Q^{(i-1)}$, $\Pbar^{(i-1)}$ and $\Qbar^{(i-1)}$ defined by
\eqref{def:Pi-1} and \eqref{def:Qi-1}, which are expanded as
\eqref{Pi-1,Qi-1:h-expand} and \eqref{Pbari-1,Qbari-1:h-expand}.

If $X$, $\Xbar$ and $\phi$ are expanded as \eqref{X:h-expansion},
\eqref{Xbar:h-expansion} and \eqref{phi:h-expansion}, the dressing
operators $W=\exp(X/\hbar)$ and
$\Wbar=\exp(\phi/\hbar)\exp(\Xbar/\hbar)$ are factorised as follows by
the Campbell-Hausdorff theorem:
\begin{align}
    W &=
    \exp\left(
     \hbar^{i-1} \tilde X_i + \hbar^i X_{>i}
    \right) 
    \exp\left( \hbar^{-1} X^{(i-1)} \right),
\label{W=exp(Xi)exp(Xi-1)}
\\
    \Wbar &=
    \exp\left(\hbar^{i-1} \phi_i + \hbar^i \phi_{>i} \right)
    \exp\left(
     \hbar^{i-1} \tilde\Xbar_i + \hbar^i \Xbar_{>i}
    \right)  \times
\label{W=exp(Xbari)exp(Xbari-1)}
\\
    &\ \ \ \times
    \exp\left(\hbar^{-1} \phi^{(i-1)} \right)
    \exp\left( \hbar^{-1} \Xbar^{(i-1)} \right),
\nonumber
\end{align}
where $\tilde X_i$, $X_{>i}$, $\phi_i$, $\phi_{>i}$, $\tilde\Xbar_i$ and
$\Xbar_{>i}$ have $\hbar$-order not more than $0$ and the principal
symbols of $\tilde X_i$ and $\tilde\Xbar_i$ are defined by
\begin{align}
    \symh(\tilde X_i){}(s,\xi)
    &=
    \sum_{n=1}^\infty \frac{(\ad_{\{,\}} \symh(X_0) )^{n-1}}{n!}
    \symh(X_i)
\label{def:tildeX}
\\
    \symh(\tilde \Xbar_i){}(s,\xi)
    &=
    e^{\ad_{\{,\}} \phi_0} \left(
    \sum_{n=1}^\infty \frac{(\ad_{\{,\}} \symh(\Xbar_0) )^{n-1}}{n!}
    \symh(\Xbar_i)
    \right)
\label{def:tildeXbar}
\end{align}
Note that the log terms in \eqref{def:tildeX} and \eqref{def:tildeXbar}
are $\alpha_i \log\xi$ and $\alphabar_i \log\xi$ respectively. The other
terms in \eqref{def:tildeX} (resp.\ \eqref{def:tildeXbar}) are negative
(resp.\ positive) powers of $\xi$. The principal symbol of $X_i$ is
recovered from $\tilde X_i$ by the formula
\begin{equation}
 \begin{split}
    \symh(X_i) 
    &= \symh(\tilde X_i)
    - \frac{1}{2} \{\symh(X_0),\symh(\tilde X_i)\}
    + \sum_{p=1}^\infty K_{2p}
      (\ad_{\{,\}} (\symh(X_0)))^{2p} \symh(\tilde X_i), 
 \end{split}
\label{tildeXi->Xi}
\end{equation}
Here coefficients $K_{2p}$ are defined by \eqref{def:K2p}. Similarly the
principal symbol of $\Xbar_i$ is recovered from $\tilde\Xbar_i$ by
\begin{equation}
 \begin{split}
    \symh(\Xbar_i) 
    &= \symh(\tilde \Xbar'_i)
    - \frac{1}{2} \{\symh(X_0),\symh(\tilde\Xbar'_i)\}
    + \sum_{p=1}^\infty K_{2p}
      (\ad_{\{,\}} (\symh(X_0)))^{2p} \symh(\tilde\Xbar'_i), 
 \end{split}
\label{tildeXbari->Xbari}
\end{equation}
where
$\symh(\tilde\Xbar'_i):=e^{-\ad_{\{,\}}\phi_0}\symh(\tilde\Xbar_i)$.

These inversion relations are the origin of
\eqref{tildeXi,Xbari->Xi,Xbari:symbol}. (Note that the principal symbol
determines the operators $X_i$ and $\Xbar_i$, since they are homogeneous
terms in the expansions \eqref{X:h-expansion} and
\eqref{Xbar:h-expansion}.) The factorisation formula
\eqref{W=exp(Xi)exp(Xi-1)} and the inversion formula \eqref{tildeXi->Xi}
are proved in Appendix A of \cite{tak-tak:09}. The formulae
\eqref{W=exp(Xbari)exp(Xbari-1)} and \eqref{tildeXbari->Xbari} are
derived in the same way.

The factorisation \eqref{W=exp(Xi)exp(Xi-1)} implies
\begin{equation*}
 \begin{split}
    P =& 
    \Ad\Bigl( 
        \exp\bigl( \hbar^{i-1} \tilde X_i + \hbar^i X_{>i} \bigr) 
       \Bigr) P^{(i-1)}
\\
    =&
    P^{(i-1)} 
    + \hbar^{i-1}
    [\tilde X_i+\hbar X_{>i}, P^{(i-1)}]
    + (\text{terms of $\hbar$-order $<-i$}).
 \end{split}
\end{equation*}
Thus, substituting the expansion \eqref{Pi-1,Qi-1:h-expand} in the Step
1, we have
\begin{equation}
 \begin{split}
    P =& 
    P^{(i-1)}_0 + \hbar P^{(i-1)}_1 + \cdots + \hbar^i P^{(i-1)}_i +
    \cdots
\\
    &+ \hbar^{i-1} [\tilde X_i, P^{(i-1)}_0]
     + (\text{terms of $\hbar$-order $<-i$}).
 \end{split}
\label{P=Ad()Pi-1}
\end{equation}
Comparing this with the $\hbar$-expansion \eqref{P} of $P$, we can
express $P_j$'s in terms of $P^{(i-1)}_j$'s and $\tilde X_i$
as follows:
\begin{align}
    P_j &= P^{(i-1)}_j \qquad  (j=0,\dots,i-1),
\label{P0-Pi-1}
\\
    \sigma_0 (P_i) &= \sigma_0 (P^{(i-1)}_i 
    + \hbar^{-1}[\tilde X_i, P^{(i-1)}_0]).
\label{Pi<-Pi-1i}
\end{align}
Similar equations for $Q$ are obtained in the same way. For the
operators $\Pbar$ the corresponding equations are
\begin{align}
    \Pbar_j &= \Pbar^{(i-1)}_j \qquad  (j=0,\dots,i-1),
\label{Pbar0-Pbari-1}
\\
    \sigma_0 (\Pbar_i) &= \sigma_0 (\Pbar^{(i-1)}_i 
    + \hbar^{-1}[\phi_i, \Pbar^{(i-1)}_0]
    + \hbar^{-1}[\tilde \Xbar_i, \Pbar^{(i-1)}_0]).
\label{Pbari<-Pbari-1i}
\end{align}
The corresponding equations for $\Qbar$ are the same.

The equations \eqref{P0-Pi-1}, \eqref{Pbar0-Pbari-1} and corresponding
equations for $Q$ and $\Qbar$ show that the terms of $\hbar$-order
greater than $-i$ in \eqref{P} are already fixed by $X_0,\dots,X_{i-1}$,
which justifies the inductive procedure. That is to say, we are assuming
that $X_0,\dots,X_{i-1}$ have been already determined so that
$P_j=P^{(i-1)}_j$ and $Q_j=Q^{(i-1)}_j$ for $j=0,\dots,i-1$ coincide
with $\Pbar_j=\Pbar^{(i-1)}_j$ and $\Qbar_j=\Qbar^{(i-1)}_j$
respectively.

The operators $X_i$, $\Xbar_i$ and the function $\phi_i$ should be
chosen so that the right hand sides of \eqref{Pi<-Pi-1i} and
\eqref{Pbari<-Pbari-1i} coincide and the corresponding expressions for
$Q$ and $\Qbar$ coincide. Taking equations $P^{(i-1)}_0=P_0$,
$Q^{(i-1)}_0=Q_0$, $\Pbar^{(i-1)}_0=\Pbar_0$ and
$\Qbar^{(i-1)}_0=\Qbar_0$ into account, we define
\begin{equation}
 \begin{split}
    \tilde P^{(i)}_i &:= P^{(i-1)}_i 
    + \hbar^{-1} [\tilde X_i, P_0],
\\
    \tilde Q^{(i)}_i &:= Q^{(i-1)}_i 
    + \hbar^{-1} [\tilde X_i, Q_0],
\\
    \tilde \Pbar^{(i)}_i &:= \Pbar^{(i-1)}_i 
    + \hbar^{-1}[\phi_i, \Pbar_0]
    + \hbar^{-1}[\tilde \Xbar_i, \Pbar_0]
\\
    \tilde \Qbar^{(i)}_i &:= \Qbar^{(i-1)}_i 
    + \hbar^{-1}[\phi_i, \Qbar_0]
    + \hbar^{-1}[\tilde \Xbar_i, \Qbar_0]
 \end{split}
\label{Pii,Qii,Pbarii,Qbarii}
\end{equation}
Then the condition for $X_i$, $\Xbar_i$ and $\phi_i$ is written in the
following form of equations for symbols:
\begin{equation}
    \symh_0(\tilde P^{(i)}_i) = \symh_0(\tilde \Pbar^{(i)}_i), \qquad
    \symh_0(\tilde Q^{(i)}_i) = \symh_0(\tilde \Qbar^{(i)}_i)
\label{Pii=Pbarii,Qii=Qbarii}
\end{equation}
(The parts of $\hbar$-order less than $-1$ should be determined in the
next step of the induction.) To simplify notations, we denote the
symbols $\symh_0(\tilde P^{(i)}_i)$, $\symh_0(P^{(i-1)}_i)$ and so on by
the corresponding calligraphic letters as $\tilde\calP^{(i)}_i$,
$\calP^{(i-1)}_i$ etc. By this notation we can rewrite the equations
\eqref{Pii=Pbarii,Qii=Qbarii} in the following form:
\begin{equation}
 \begin{aligned}
    \tilde \calP^{(i)}_i &= \tilde \calPbar^{(i)}_i,&
    \tilde \calQ^{(i)}_i &= \tilde \calQbar^{(i)}_i,
\\
    \tilde \calP^{(i)}_i &:= \calP^{(i-1)}_i 
    + \{ \tilde \calX_i, \calP_0\},&
    \tilde \calQ^{(i)}_i &:= \calQ^{(i-1)}_i 
    + \{\tilde \calX_i, \calQ_0\},
\\
    \tilde \calPbar^{(i)}_i &:= \calPbar^{(i-1)}_i 
    + \{\phi_i, \calPbar_0\}
    + \{ \tilde \calXbar_i, \calPbar_0\},&
    \tilde \calQbar^{(i)}_i &:= \calQbar^{(i-1)}_i 
    + \{ \phi_i, \calQbar_0\}
    + \{ \tilde \calXbar_i, \calQbar_0\}.
  \end{aligned}
\label{Pii,Qii:symbol}
\end{equation}
In the matrix form, these equations are encapsulated in the following
equation. 
\begin{equation}
 \begin{split}
    &
    \begin{pmatrix} 
     \calP^{(i-1)}_i \\\\ \calQ^{(i-1)}_i
    \end{pmatrix}
    +
    \xi
    \begin{pmatrix}
     \dfrac{\der \calP_0}{\der s} & - \dfrac{\der \calP_0}{\der \xi} \\ \\
     \dfrac{\der \calQ_0}{\der s} & - \dfrac{\der \calQ_0}{\der \xi} 
    \end{pmatrix}
     \begin{pmatrix}
     \dfrac{\der}{\der \xi} \tilde\calX_i\\ \\
     \dfrac{\der}{\der s} \tilde\calX_i
     \end{pmatrix}
\\
    ={}&
    \begin{pmatrix} 
     \calPbar^{(i-1)}_i \\\\ \calQbar^{(i-1)}_i
    \end{pmatrix}
    +
    \xi
    \begin{pmatrix}
     \dfrac{\der \calPbar_0}{\der s} 
                           & - \dfrac{\der \calPbar_0}{\der \xi} \\ \\
     \dfrac{\der \calQbar_0}{\der s} 
                           & - \dfrac{\der \calQbar_0}{\der \xi} 
    \end{pmatrix}
     \begin{pmatrix}
     \dfrac{\der}{\der \xi} (\phi_i + \tilde\calXbar_i)\\ \\
     \dfrac{\der}{\der s}   (\phi_i + \tilde\calXbar_i)
     \end{pmatrix}
\end{split}
\label{Pii,Qii:mat}
\end{equation}
Recall that operators $P^{(i-1)}$ and $Q^{(i-1)}$ are defined by acting
adjoint operation to the canonically commuting pair $(f,g)$ in
\eqref{def:Pi-1}, \eqref{def:Qi-1} and \eqref{ft,gt}. Hence they also
satisfy the canonical commutation relation: $[P^{(i-1)},Q^{(i-1)}]=\hbar
P^{(i-1)}$. The principal symbol of this relation gives
\begin{equation}
   \{\calP^{(i-1)}_0,\calQ^{(i-1)}_0\}=\{\calP_0,\calQ_0\}
   =
   \calP_0,
\label{CCR(P0,Q0)} 
\end{equation}
which means that the determinants of the $2\times2$ matrices in both
sides of \eqref{Pii,Qii:mat} are equal to $\xi^{-1}\calP_0$. (Recall
that those matrices are equal because of the induction hypothesis,
$\calP_0=\calPbar_0$, $\calQ_0=\calQbar_0$.) Hence
its inverse matrix is easily computed and we have
\begin{equation}
 \begin{split}
    &\calP_0^{-1}
    \begin{pmatrix}
    - \dfrac{\der \calQ_0}{\der \xi}  & \dfrac{\der \calP_0}{\der \xi}
    \\ \\
    - \dfrac{\der \calQ_0}{\der s}    & \dfrac{\der \calP_0}{\der s}
    \end{pmatrix}
    \begin{pmatrix} 
     \calP^{(i-1)}_i \\\\ \calQ^{(i-1)}_i
    \end{pmatrix}
    +
    \begin{pmatrix}
     \dfrac{\der}{\der \xi} \tilde\calX_i\\ \\
     \dfrac{\der}{\der s} \tilde\calX_i
    \end{pmatrix}
\\
    ={}&
    \calPbar_0^{-1}
    \begin{pmatrix}
    - \dfrac{\der \calQbar_0}{\der \xi}
                     & \dfrac{\der \calPbar_0}{\der \xi}
    \\ \\
    - \dfrac{\der \calQbar_0}{\der s}
                     & \dfrac{\der \calPbar_0}{\der s}
    \end{pmatrix}
    \begin{pmatrix} 
     \calPbar^{(i-1)}_i \\\\ \calQbar^{(i-1)}_i
    \end{pmatrix}
    +
     \begin{pmatrix}
     \dfrac{\der}{\der \xi} (\phi_i + \tilde\calXbar_i)\\ \\
     \dfrac{\der}{\der s}   (\phi_i + \tilde\calXbar_i)
     \end{pmatrix}
 \end{split}
\label{P+Xi'=Pbar+phii'+Xbari':mat}
\end{equation}
Note that the left hand side (the first line) is a series of $\xi$
around $\xi=\infty$, while the right hand side (the second line) is a
series around $\xi=0$. Equation \eqref{P+Xi'=Pbar+phii'+Xbari':mat} is
rewritten as
\begin{multline}
    \begin{pmatrix}
    \dfrac{\der}{\der \xi}(-\tilde\calX_i+\phi_i+\tilde\calXbar_i) \\ \\
    \dfrac{\der}{\der s}  (-\tilde\calX_i+\phi_i+\tilde\calXbar_i) 
    \end{pmatrix}
\\
    =
    \calP_0^{-1}
    \begin{pmatrix}
    - \dfrac{\der \calQ_0}{\der \xi}  & \dfrac{\der \calP_0}{\der \xi}
    \\ \\
    - \dfrac{\der \calQ_0}{\der s}    & \dfrac{\der \calP_0}{\der s}
    \end{pmatrix}
    \begin{pmatrix}
    \calP^{(i-1)}_i \\ \\
    \calQ^{(i-1)}_i
    \end{pmatrix}
    -
    \calPbar_0^{-1}
    \begin{pmatrix}
    - \dfrac{\der \calQbar_0}{\der \xi}
                       & \dfrac{\der \calPbar_0}{\der \xi}
    \\ \\
    - \dfrac{\der \calQbar_0}{\der s}
                       & \dfrac{\der \calPbar_0}{\der s}
    \end{pmatrix}
    \begin{pmatrix}
    \calPbar^{(i-1)}_i \\ \\
    \calQbar^{(i-1)}_i
    \end{pmatrix},
\label{-Xi+phii+Xbari:mat}
\end{multline}
which determines $-\tilde\calX_i+\phi_i+\tilde\calXbar_i$. According to
the above remark, the first term in the right hand side is a series of
$\xi$ around $\xi=\infty$ and the second term is a series of $\xi$
around $\xi=0$.

The system \eqref{-Xi+phii+Xbari:mat} is solvable thanks to
\lemref{lem:compatibility} below. Hence, integrating the first element
of the right hand side with respect to $\xi$, we obtain
$-\tilde\calX_i+\phi_0+\tilde\calXbar_i$ up to an integration constant
which does not depend on $\xi$. Integrating the second element of the
above equation, we can determine this integration constant up to a
constant which does not depend on $s$. This is Step 2,
\eqref{tildeXi=intxi} and \eqref{tildeXi=ints}, which determine the
symbols $\tilde\calX_i$, $\tilde\calXbar_i$ and the function
$\phi_i$. (The ambiguity of $\phi_i$ is harmless.)

In the end, the principal symbols of $X_i$ and $\Xbar_i$ are determined
by \eqref{tildeXi->Xi}, \eqref{tildeXbari->Xbari} or
\eqref{tildeXi,Xbari->Xi,Xbari:symbol} in Step 3. Operators $X_i$ and
$\Xbar_i$ are defined as in Step 4. This completes the construction of
$X_i$, $\Xbar_i$ and $\phi_i$ and the proof of the theorem.
%

\begin{lem}
\label{lem:compatibility}
 The system \eqref{-Xi+phii+Xbari:mat} is compatible.
\end{lem}

\begin{proof}
 We check the compatibility condition,
\begin{equation}
 \begin{split}
    &\frac{\der}{\der s} \left(
    \calP_0^{-1}
    \left(
    -
    \dfrac{\der \calQ_0}{\der \xi} \calP^{(i-1)}_i
    +
    \dfrac{\der \calP_0}{\der \xi} \calQ^{(i-1)}_i
    \right)
    \right)
\\
    -
    &\frac{\der}{\der s} \left(
    \calPbar_0^{-1}
    \left(
    -
    \dfrac{\der \calQbar_0}{\der \xi} \calPbar^{(i-1)}_i
    +
    \dfrac{\der \calPbar_0}{\der \xi} \calQbar^{(i-1)}_i
    \right)
    \right)
\\
    =&
    \frac{\der}{\der \xi} \left(
    \calP_0^{-1}
    \left(
    -
    \dfrac{\der \calQ_0}{\der s} \calP^{(i-1)}_i
    +
    \dfrac{\der \calP_0}{\der s} \calQ^{(i-1)}_i
    \right)
    \right)
\\
    -&
    \frac{\der}{\der \xi}\left(
    \calPbar_0^{-1}
    \left(
    -
    \dfrac{\der \calQbar_0}{\der s} \calPbar^{(i-1)}_i
    +
    \dfrac{\der \calPbar_0}{\der s} \calQbar^{(i-1)}_i
    \right)
    \right).
 \end{split}
\label{compatibility}
\end{equation}
 Using the relation $\{\calP_0,\calQ_0\}=\calP_0$ \eqref{CCR(P0,Q0)},
 this equation reduces to
\begin{equation}
 \begin{split}
     &\calP_0^{-1} \xi^{-1}
    \bigl(
     - \calP^{(i-1)}_i
     + \{ \calP^{(i-1)}_i, \calQ_0 \}
     + \{ \calP_0, \calQ^{(i-1)}_i\}
    \bigr)
 \\
    =
    &\calPbar_0^{-1} \xi^{-1}
    \bigl(
     - \calPbar^{(i-1)}_i
     + \{ \calPbar^{(i-1)}_i, \calQbar_0 \}
     + \{ \calPbar_0, \calQbar^{(i-1)}_i \}
    \bigr).
 \end{split}
\label{compatibility2}
\end{equation}
 Defined from canonically commuting pair $(f,g)$ by adjoint action
 \eqref{ft,gt}, \eqref{def:Pi-1} and \eqref{def:Qi-1}, the pair of
 operators $(P^{(i-1)}, Q^{(i-1)})$ is canonically commuting:
 $[P^{(i-1)},Q^{(i-1)}]=\hbar P^{(i-1)}$. Similarly we have
 $[\Pbar^{(i-1)},\Qbar^{(i-1)}]=\hbar \Pbar^{(i-1)}$ and thus
\begin{equation}
    [P^{(i-1)},     Q^{(i-1)}    ] - \hbar P^{(i-1)}
    =
    [\Pbar^{(i-1)}, \Qbar^{(i-1)}] - \hbar \Pbar^{(i-1)}.
\label{ccr-ccrbar}
\end{equation}
 Substituting the expansions \eqref{Pi-1,Qi-1:h-expand} and
 \eqref{Pbari-1,Qbari-1:h-expand} in it and noting that
 $P^{(i-1)}_j=\Pbar^{(i-1)}_j$ and $Q^{(i-1)}_j=\Qbar^{(i-1)}_j$ for
 $j=0,\dots,i-1$ by the induction hypothesis, the terms of $\hbar$-order
 higher than $-i-1$ in \eqref{ccr-ccrbar} cancel. Thus
 \eqref{ccr-ccrbar} becomes 
\begin{equation*}
 \begin{split}
    &[\hbar^i P^{(i-1)}_i, Q_0] + [P_0, \hbar^i Q^{(i-1)}_i]
     -
     \hbar^{i+1} P^{(i-1)}_i
     + (\text{$\hbar$-order $< -i-1$})
\\
    ={}&
     [\hbar^i \Pbar^{(i-1)}_i, \Qbar_0]
     +
     [\Pbar_0, \hbar^i \Qbar^{(i-1)}_i]
     - 
     \hbar^{i+1} \Pbar^{(i-1)}_i
     + (\text{$\hbar$-order $< -i-1$}).
 \end{split}
\end{equation*}
 Taking the symbol of $\hbar$-order $-i-1$ of this equation, we have
 \eqref{compatibility2} because $\calP_0=\calPbar_0$. 
\end{proof}

\section{Asymptotics of the wave function}
\label{sec:wave-function}

In this section we prove that the dressing operator of the form
\eqref{W=exp(X)} or \eqref{W=exp(Xcirc)exp(alpha)} (with $\alpha=0$),
i.e.,
\begin{align}
    W(\hbar,t,\tbar,s, e^{\hbar\ders})
    &=
    \exp(X(\hbar,t,\tbar,s,e^{\hbar\ders})/\hbar),
\label{W=exp(X):alpha=0}
\\
    X(\hbar,t,\tbar,s,e^{\hbar\ders})
    &=
    \sum_{k=1}^\infty \chi_k   (\hbar,t,\tbar,s) e^{-k\hbar\ders},
    \qquad 
    \ordh X \leqq 0,
\label{X:alpha=0}
\end{align}
gives the wave function of the WKB form
\begin{gather}
    \Psi(\hbar,t,\tbar,s;z) 
    = W z^{s/\hbar} e^{\zeta(t,z)/\hbar} 
    = e^{S(\hbar,t,\tbar,s,z)/\hbar} z^{s/\hbar},
    \qquad \ordh S \leqq 0,
\label{wave-func}
\\
    S(\hbar,t,\tbar,s;z)
    = \sum_{n=0}^\infty \hbar^n S_n(t,\tbar,s;z) + \zeta(t,z), \qquad
    \zeta(t,z) := \sum_{n=1}^\infty t_n z^n,
\label{S}
\end{gather}
and vice versa. As the factor $e^{\hbar^{-1}\alpha(\hbar)(\hbar\ders)}$
in \eqref{W=exp(Xcirc)exp(alpha)} becomes a constant factor
$z^{\alpha(\hbar)/\hbar}$ when it is applied to $z^{s/\hbar}$, we omit
it here.

By changing the sign of $s$ and replacing $z$ by $\zbar^{-1}$, we can
deduce the formula for the wave function $\Psibar$ corresponding to the
dressing operator $\Wbar$ of the form \eqref{Wbar=exp(phi)exp(Xbar)} or
\eqref{Wbar=exp(phi)exp(Xbarcirc)exp(alpha)} (with $\alphabar=0$)
immediately from the above results: if $\Wbar$ has the form
\begin{equation}
 \begin{aligned}
    \Wbar(\hbar,t,\tbar,s, e^{\hbar\ders})
    &=
    \exp(\phi(\hbar,t,\tbar,s)/\hbar)
    \exp(\Xbar(\hbar,t,\tbar,e^{\hbar\ders})/\hbar),
\\
    \Xbar(\hbar,t,\tbar,s,e^{\hbar\ders})
    &=
    \sum_{k=1}^\infty \chibar_k   (\hbar,t,\tbar,s) e^{k\hbar\ders},
    \qquad 
    \ordh \phi, \ordh \Xbar \leqq 0,
 \end{aligned}
\label{Wbar=exp(Xbar):alphabar=0}
\end{equation}
gives the wave function of the WKB form
\begin{gather}
    \Psibar(\hbar,t,\tbar,s;\zbar) 
    = \Wbar  \zbar^{s/\hbar} e^{\zeta(\tbar,\zbar^{-1})/\hbar} 
    = e^{\Sbar(\hbar,t,\tbar,s,\zbar)/\hbar} \zbar^{s/\hbar},
    \qquad \ordh S \leqq 0,
\label{wave-func:bar}
\\
    \Sbar(\hbar,t,\tbar,s;\zbar)
    = \sum_{n=0}^\infty \hbar^n \Sbar_n(t,\tbar,s;\zbar) 
    + \zeta(\tbar,\zbar^{-1}), \qquad
    \zeta(\tbar,\zbar^{-1}) := \sum_{n=1}^\infty \tbar_n \zbar^{-n}.
\label{Sbar}
\end{gather}
Since the time variables $t_n$ and $\tbar_n$ do not play any role in
this section, we set them to zero.

Let $A(\hbar,s,e^{\hbar\ders})=\sum_n a_n(\hbar,s) e^{n\hbar\ders}$ be a
difference operator. The {\em total symbol} of $A$ is a power
series of $\xi$ defined by
\begin{equation}
    \totsym(A)(\hbar,s,\xi):= \sum_n a_n(\hbar,x) \xi^n.
\label{tot-symbol}
\end{equation}
Actually, this is the factor which appears when the operator $A$ is
applied to $z^{s/\hbar}$:
\begin{equation}
    A z^{s/\hbar} = \totsym(A)(\hbar,s,z) z^{s/\hbar}.
\label{tot-symbol:Laplace}
\end{equation}
Using this terminology, what we show in this section is that a operator
of the form $e^{X/\hbar}$ has a total symbol of the form $e^{S/\hbar}$
and that an operator with total symbol $e^{S/\hbar}$ has a form
$e^{X/\hbar}$. Exactly speaking, the main results in this section are
the following two propositions.
\begin{prop}
\label{prop:exp(:X:)=:exp(S):}
 Let $X=X(\hbar,s,e^{\hbar\ders})$ be a difference operator of the form
 \eqref{X:alpha=0}, which has the $\hbar$-order $0$: $\ordh X = 0$. Then
 the total symbol of $e^{X/\hbar}$ has such a form as
\begin{equation}
    \totsym(\exp(\hbar^{-1} X(\hbar,s,e^{\hbar\ders})))
    = e^{S(\hbar,s,\xi)/\hbar},
\label{exp(:X:)=:exp(S):}
\end{equation}
 where $S(\hbar,s,\xi)$ is a power series of $\xi^{-1}$ without
 non-negative powers of $\xi$ and has an $\hbar$-expansion
\begin{equation*}
    S(\hbar,s,\xi)=\sum_{n=0}^\infty \hbar^n S_n(s,\xi).
\end{equation*}

 Moreover, the coefficient $S_n$ is determined by $X_0,\dots,X_n$ in the
 $\hbar$-expansion \eqref{X:h-expansion} of $X=\sum_{n=0}^\infty \hbar^n
 X_n$.
\end{prop}

Explicitly, $S_n$ is determined as follows:
\begin{itemize}
 \item (Step 0) Assume that $X_0,\dots,X_n$ are given. Let $X_i(s,\xi)$
       be the total symbol $\totsym(X_i(s,e^{\hbar\ders}))$.

 \item (Step 1) Define $Y^{(l)}_{k,m}(s,s',\xi,\xi')$ and
       $S^{(l)}(s,\xi)$ by the following recursion relations:
\begin{align}
    &Y^{(l)}_{k,-1}=0
\\
    &S^{(0)}_m=0,
\\
    &Y^{(l)}_{0,m}(s,s',\xi,\xi') = \delta_{l,0} X_m(s,\xi)
\label{Yl0m}
\end{align}
for $l\geqq 0$, $m=0,\dots,n$,
\begin{multline}
    Y^{(l)}_{k+1,m}(s,s',\xi,\xi')
\\    =
    \frac{1}{k+1}
    \left(
    \xi\der_\xi \der_{s'} Y^{(l)}_{k,m-1}(s,s',\xi,\xi')
    +
    \sum_{\substack{0\leq l' \leq l-1 \\ 0\leq m' \leq m}}
    \xi\der_\xi Y^{(l')}_{k,m'}(s,s',\xi,\xi')
    \der_{s'} S^{(l-l')}_{m-m'}(s',\xi')
    \right)
\label{recursion:Y(l,k,m)}
\end{multline}
       for $k\geqq 0$, and
\begin{equation}
    S^{(l+1)}_m(s,\xi)
    = \frac{1}{l+1} \sum_{k=0}^{l+m} Y^{(l)}_{k,m}(s,s,\xi,\xi).
\label{recursion:S(l+1,m)}
\end{equation}
       (We shall prove that $Y^{(l)}_{k,m}=0$ if $k>l+m$.) Schematically
       this procedure goes as follows:
{\small
\begin{equation*}
 \begin{matrix}
                     &         & Y^{(l)}_{0,0}=\delta_{l,0}X_0
                     &         & Y^{(l)}_{0,1}=\delta_{l,0}X_1
                     &         & Y^{(l)}_{0,2}=\delta_{l,0}X_2
\\
                     &         & +
                     &\searrow & +
                     &\searrow & +
\\
    Y^{(l)}_{k,-1}=0 &\to      & Y^{(l)}_{k,0}
                     &\to      & Y^{(l)}_{k,1}
                     &\to      & Y^{(l)}_{k,2} & \cdots
\\
                     &         & \downarrow
                     &\nearrow & \downarrow
                     &\nearrow & \downarrow
\\
                     &         & S^{(l+1)}_0
                     &         & S^{(l+1)}_1
                     &         & S^{(l+1)}_2
 \end{matrix}
\end{equation*}
}

 \item (Step 2) $S_n(s,\xi)= \sum_{l=1}^\infty S^{(l)}_n(s,\xi)$. (The
       sum makes sense as a power series of $\xi$.)
\end{itemize}

\begin{prop}
\label{prop::exp(S):=exp(:X:)}
 Let $S(\hbar,s,\xi)=\sum_{n=0}^\infty \hbar^n S_n(s,\xi)$ be a power
 series of $\xi^{-1}$ without non-negative powers of $\xi$. Then there
 exists a difference operator $X(\hbar,s,e^{\hbar\ders})$ of the form
 \eqref{X:alpha=0} such that $\ordh X \leqq 0$ and
\begin{equation}
    \totsym(\exp(\hbar^{-1} X(\hbar,s,e^{\hbar\ders})))
    = e^{S(\hbar,s,\xi)/\hbar}.
\label{:exp(S):=exp(:X:)}
\end{equation}
 Moreover, the coefficient $X_n(s,\xi)$ in the $\hbar$-expansion
 $X=\sum_{n=0}^\infty \hbar^n X_n$ of the total symbol
 $X=X(\hbar,s,\xi)$ is determined by $S_0,\dots,S_n$ in the
 $\hbar$-expansion of $S$.
\end{prop}

Explicit procedure is as follows:
\begin{itemize}
 \item (Step 0) Assume that $S_0,\dots,S_n$ are given. Expand them into
       homogeneous terms with respect to powers of $\xi$:
       $S_n(s,\xi)=\sum_{j=1}^\infty S_{n,j}(s,\xi)$, where $S_{n,j}$ is
       a term of degree $-j$.
 \item (Step 1) Define $Y^{(l)}_{k,n,j}(s,s',\xi,\xi')$ as follows:
\begin{align}
    &Y^{(l)}_{k,-1,j}(s,s',\xi,\xi')=0,
\\
    &Y^{(l)}_{k,m,1}(s,s',\xi,\xi')
    =\delta_{l,0}\delta_{k,0} S_{m,1}(s,\xi)
\label{Ylkm1}
\end{align}
       for $m=0,\dots,n$, $k\geqq 0$, $l\geqq 0$ and
\begin{equation}
    Y^{(l)}_{0,m,j}=0
\end{equation}
       for $m=0,\dots,n$, $l>0$, $j\geqq 1$. For other $(l,k,m,j)$,
       $(l,k)\neq(0,0)$, $Y^{(l)}_{k,m,j}$ are determined by the
       recursion relation:
\begin{multline}
    Y^{(l)}_{k+1,m,j}(s,s',\xi,\xi')
    =
    \frac{1}{k+1}
    \Biggl(
    \xi\der_\xi \der_{s'} Y^{(l)}_{k,m-1,j}(s,s',\xi,\xi')+
\\
    +
    \sum_{\substack{0\leq l' \leq l-1\\
                    1\leq j' \leq j-1, 
                    0\leq m' \leq m \\
                    0\leq k'' \leq l-l'-1+m-m'}}
    \frac{1}{l-l'}
    \der_\xi Y^{(l')}_{k,m',j'}(s,s',\xi,\xi')
    \der_y Y^{(l-l'-1)}_{k'',m-m',j-j'}(s,s,\xi,\xi)
    \Biggr).
\label{recursion:Y(l,k,m)<-Y}
\end{multline}
       The remaining $Y^{(0)}_{0,m,j}$ is determined by:
\begin{equation}
    Y^{(0)}_{0,m,j}(s,s',\xi,\xi')
    =
    S_{m,j}(s,\xi)
    - \sum_{\substack{(l,k)\neq(0,0)\\
                       0 \leq l < j, 0 \leq k \leq l+m,}
           }
      \frac{1}{l+1} Y^{(l)}_{k,m,j}(s,s,\xi,\xi).
\label{Y00mj=Smj-Ylkmj}
\end{equation}
       (We shall show that $Y^{(l)}_{k,m,j}=0$ for $k>l+m$ or $j \leqq
       l$.)  Schematically this procedure goes as follows:
\begin{equation*}
 \begin{matrix}
    &  & 
    Y^{(l)}_{k,m,1}=\delta_{l,0}\delta_{k,0}S_{m,1}
\\
    &  & 
    \downarrow\hskip 1.5cm 
\\
    Y^{(l')}_{k',m',1}(m'<m)      & \to  & 
    Y^{(l)}_{k,m,2}\ (k,l\neq 0)  & \to  &
    Y^{(0)}_{0,m,2}               & \leftarrow &
    S_{m,2}
\\
    &  & 
    \downarrow\hskip 1.5cm & \swarrow
\\
    Y^{(l')}_{k',m',1}, Y^{(l')}_{k',m',2}(m'<m)      & \to  & 
    Y^{(l)}_{k,m,3}\ (k,l\neq 0)  & \to  &
    Y^{(0)}_{0,m,3}               & \leftarrow &
    S_{m,3}
\\
    &  & 
    \vdots\hskip 1.5cm
 \end{matrix}
\end{equation*}

\item (Step 2) $X_n(s,\xi)=\sum_{j=1}^\infty
       Y^{(0)}_{0,n,j}(s,s,\xi,\xi)$.  (The infinite sum is the
       homogeneous expansion in terms of powers of $\xi$.)
\end{itemize}

Combining these propositions (and the corresponding statements for
$\Xbar$ and $\Sbar$ \eqref{Sbar}) with the results in
\secref{sec:recursion}, we can, in principle, make a recursion formula
for $S_n$ and $\Sbar_n$ ($n=0,1,2,\dots$) of the wave functions of the
solution of the Toda lattice hierarchy corresponding to
$(f,g,\fbar,\gbar)$ by \propref{prop:RH} (i) as follows: let
$S_0,\dotsc,S_{i-1}$, $\Sbar_0,\dotsc,\Sbar_{i-1}$ and
$\phi_0,\dotsc,\phi_{i-1}$ be given.
\begin{enumerate}
 \item By \propref{prop::exp(S):=exp(:X:)} and its variant with the
       opposite sign of $s$ we have $X_0,\dotsc,X_{i-1}$ and
       $\Xbar_0,\dotsc,\Xbar_{i-1}$.
 \item We have a recursion formula for $X_i$, $\Xbar_i$ and $\phi_i$ by
       \thmref{thm:recursion:X,Xbar}. 
 \item \propref{prop:exp(:X:)=:exp(S):} (with its variant) gives a
       formula for $S_i$, $\Sbar_i$.
\end{enumerate}
If we take the factor $e^{\hbar^{-1}\alpha(\hbar)(\hbar\ders)}$ into
account, this process becomes a little bit more complicated, but
essentially the same.

\bigskip
The rest of this section is devoted to the proof of
\propref{prop:exp(:X:)=:exp(S):} and \propref{prop::exp(S):=exp(:X:)}. 

To avoid confusion, the commutative multiplication of total symbols
$a(\hbar,s,\xi)$ and $b(\hbar,s,\xi)$ as power series is denoted by
$a(\hbar,s,\xi)\, b(\hbar,s,\xi)$ and the non-commutative multiplication
corresponding to the operator product is denoted by $a(\hbar,s,\xi)
\circ b(\hbar,s,\xi)$. Recall that the latter multiplication is
expressed (or defined) as follows:
\begin{equation}
 \begin{split}
    a(\hbar,s,\xi) \circ b(\hbar,s,\xi)
    &=
    e^{\hbar\, \xi\der_\xi \der_{s'}} 
    a(\hbar,s,\xi) b(\hbar,s',\xi')|_{s'=s,\xi'=\xi}
\\
    &=
    \sum_{n=0}^\infty\frac{\hbar^n}{n!}
    (\xi\der_\xi)^n a(\hbar,s,\xi) 
        \der_{s'}^n b(\hbar,s',\xi')|_{s'=s,\xi'=\xi}.
\label{def:symbol-prod}
 \end{split}
\end{equation}
(This corresponds to Equation (3.21) of \cite{tak-tak:09} for
microdifferential operators.)  The order of the symbol $a(\hbar,s,\xi)$
(the order with respect to $\xi$ as a power series of $\xi$) is denoted
by $\ordx a(\hbar,s,\xi)$:
\begin{equation*}
    \ordx \left( \sum a_{m}(\hbar,s) \xi^{m} \right)
    \defeq
    \max \left\{ m \,\left|\, 
                 a_m(\hbar,s) \neq 0 
    \right.\right\}.
\end{equation*}
The $\hbar$-order is the same as that of 
operators: $\ordh s = \ordh \xi = 0$, $\ordh \hbar = -1$.

The main idea of proof of propositions is the same as those in \S3 of
\cite{tak-tak:09}, which is a formal version of Aoki's exponential
calculus of microdifferential operators, \cite{aok:86}. Since the Euler
operator $\xi\der_\xi$ does not lower the order with respect to
$\xi$ in contrast to the differential operator $\der_\xi$, proof of
convergence of series like \eqref{def:symbol-prod} as a formal power
series is different from that in \cite{tak-tak:09}.

\bigskip
First, we prove the following lemma.
\begin{lem}
\label{lem:prod-exp-symbol}
 Let $a(\hbar,s,\xi)$, $b(\hbar,s,\xi)$, $p(\hbar,s,\xi)$ and
 $q(\hbar,s,\xi)$ be symbols such that $\ordx a(\hbar,s,\xi) = M$,
 $\ordh a(\hbar,s,\xi)\leqq 0$, $\ordx b(\hbar,s,\xi) = N$, $\ordh
 b(\hbar,s,\xi)\leqq 0$, $\ordx p(\hbar,s,\xi)\leqq -1$, $\ordx
 q(\hbar,s,\xi)\leqq -1$, $\ordh p(\hbar,s,\xi)\leqq 0$, $\ordh
 q(\hbar,s,\xi)\leqq 0$.

 Then there exist symbols $c(\hbar,s,\xi)$ ($\ordx c(\hbar,s,\xi) =
 N+M$, $\ordh c(\hbar,s,\xi) \leqq 0$) and $r(\hbar,s,\xi)$ ($\ordx
 r(\hbar,s,\xi) \leqq 0$, $\ordh r(\hbar,s,\xi) \leqq 0$) such that
\begin{equation}
    \bigl(a(\hbar,s,\xi) e^{p(\hbar,s,\xi)/\hbar}\bigr)
    \circ
    \bigl(b(\hbar,s,\xi) e^{q(\hbar,s,\xi)/\hbar}\bigr)
    =
    c(\hbar,s,\xi) e^{r(\hbar,s,\xi)/\hbar}.
\label{prod-exp-symbol}
\end{equation}
\end{lem}

In the proof of \propref{prop:exp(:X:)=:exp(S):} and
\propref{prop::exp(S):=exp(:X:)}, we use the construction of $c$ and $r$
in the proof of \lemref{lem:prod-exp-symbol}.

\begin{proof}
 Following \cite{aok:86}, we introduce a parameter $t$ and consider
\begin{equation}
    \pi(t) = \pi(t;\hbar,s,s',\xi,\xi'):=
    e^{\hbar t \xi\der_\xi \der_{s'}}
    a(\hbar,s,\xi) b(\hbar,s',\xi')
    e^{\bigl(p(\hbar,s,\xi) + q(\hbar,s',\xi')\bigr)/\hbar}.
\label{def:pi(t)}
\end{equation}
 If we set $t=1$, $s'=s$ and $\xi'=\xi$, this reduces to the operator
 product of \eqref{def:symbol-prod}. The series $\pi(t)$ is the unique
 solution of an initial value problem:
\begin{equation}
    \der_t \pi = \hbar \xi\der_\xi\der_{s'} \pi, \qquad
    \pi(0) = a(\hbar,s,\xi) b(\hbar,s',\xi') 
    e^{\bigl(p(\hbar,s,\xi) + q(\hbar,s',\xi')\bigr)/\hbar}.
\label{diff-eq:pi}
\end{equation}
 We construct its solution in the following form:
\begin{equation}
 \begin{aligned}
    \pi(t) &= \psi(t) e^{w(t)/\hbar},
\\
    \psi(t) &= \psi(t;\hbar,s,s',\xi,\xi')
             = \sum_{n=0}^\infty \psi_n t^n,
\\
    w(t) &= w(t;\hbar,s,s',\xi,\xi') = \sum_{k=0}^\infty w_k t^k.
 \end{aligned}
\label{pi=psi*ew}
\end{equation}
 Later we set $t=1$ and prove that $\psi(1)$ and $w(1)$ are meaningful
 as a formal power series of $\xi$ and $\xi'$. The differential equation
 \eqref{diff-eq:pi} is rewritten as
\begin{equation}
 \begin{split}
    &\frac{\der\psi}{\der t} + \hbar^{-1} \psi \frac{\der w}{\der t}
\\
    =&
    \hbar \xi\der_\xi \der_{s'} \psi 
    + \xi\der_\xi \psi \der_{s'} w + \xi\der_\xi w \der_{s'} \psi 
    +
    \psi 
    \left(              \xi\der_\xi \der_{s'} w 
         + \hbar^{-1} \xi\der_\xi w \der_{s'} w
    \right).
 \end{split}
\label{diff-eq:psi,w}
\end{equation}
 Hence it is sufficient to construct
 $\psi(t)=\psi(t;\hbar,s,s',\xi,\xi')$ and $w(t)=w(t;\hbar,s,s',\xi,\xi')$
 which satisfy 
\begin{align}
    \frac{\der w}{\der t} &= \hbar \xi\der_\xi \der_{s'} w
    + \xi\der_\xi w \der_{s'} w,
\label{diff-eq:w}
\\
    \frac{\der \psi}{\der t}
    &= \hbar \xi\der_\xi \der_{s'} \psi 
     + \xi\der_\xi \psi \der_{s'} w + \xi\der_\xi w \der_{s'} \psi.
\label{diff-eq:psi}
\end{align}
 (This is a {\em sufficient} condition but not a necessary condition for
 $\pi = \psi e^{w/\hbar}$ to be a solution of \eqref{diff-eq:pi}. The
 solution of \eqref{diff-eq:pi} is unique, but $\psi$ and $w$ satisfying
 \eqref{diff-eq:psi,w} are not unique at all.)

 To begin with, we solve \eqref{diff-eq:w} and determine
 $w(t)$. Expanding it as $w(t) = \sum_{k=0}^\infty w_k t^k$, we have
 a recursion relation and the initial condition
\begin{equation}
 \begin{split}
    w_{k+1} &= \frac{1}{k+1} \left(
    \hbar \xi\der_\xi \der_{s'} w_k +
    \sum_{\nu=0}^k \xi\der_\xi w_\nu \der_{s'} w_{k-\nu}
    \right),
\\
    w_0 &= p(s,\xi) + q(s',\xi'),
 \end{split}
\label{recursion:w}
\end{equation}
 which determine $w_k=w_k(\hbar,s,s',\xi,\xi')$ inductively. In order to
 show that $\sum_{k=0}^\infty w_k$ converges as a formal power series,
 let us expand each $w_k$ as follows:
\begin{equation}
    w_k(\hbar,s,s',\xi,\xi') 
    =
    \sum_{n=0}^\infty \hbar^n w_{k,n}(s,s',\xi,\xi').
\label{wk:h-expansion}
\end{equation}
 Expanding \eqref{recursion:w} as a series of $\hbar$, we obtain a
 recursion relation of $w_{k,n}$ and the initial condition
\begin{equation}
 \begin{split}
    w_{k+1,n} &= \frac{1}{k+1} \left(
    \xi\der_\xi \der_{s'} w_{k,n-1} +
    \sum_{\substack{k'+k''=k\\n'+n''=n}}
     \xi\der_\xi w_{k',n'} \der_{s'} w_{k'',n''}
    \right),
\\
    w_0 &= p(s,\xi) + q(s',\xi').
 \end{split}
\label{recursion:wkn}
\end{equation}
 Because of the assumption $\ordx p \leqq -1$ and $\ordx q \leqq -1$,
 $w_0$ also has the order $\leqq -1$ and consequently
\begin{equation}
     \ordx w_{0,n} \leqq -1.
\label{ordx(w0n)=-1}
\end{equation}
 (Here $\ordx$ means the order with respect to both $\xi$ and $\xi'$:
$
    \ordx \left( \sum a_{m,n}(\hbar,s) \xi^{m}\xi'^{n} \right)
\allowbreak
    \defeq
    \max \left\{ m+n \,\left|\, 
                 a_{m,n}(\hbar,s) \neq 0 
    \right.\right\}.
$)
 We show 
\begin{equation}
    \ordx w_{k,n} \leqq \min(-1,-k+n-1)
\label{ord(wkn)=-k+n-1}
\end{equation}
 by induction on $k$. 
\begin{itemize}
 \item First, when $k=0$, \eqref{ord(wkn)=-k+n-1} holds for any $n\geqq
       0$ because of \eqref{ordx(w0n)=-1}.

 \item Assume that \eqref{ord(wkn)=-k+n-1} holds for any pair $(k,n)$
       with $n\geqq 0$ and $k=0,\dots,k_0$. Then the right hand side of
       \eqref{recursion:wkn} with $k=k_0$ has the order (with respect to
       $\xi$ and $\xi'$) less than or equal to $-1$ and 
       $-k_0+(n-1)-1=(-k'+n'-1)+(-k''+n''-1)=-k_0+n-2$, since
       $\xi\der_\xi$ does not change the order.

 \item Hence \eqref{ord(wkn)=-k+n-1} is true for $k=k_0+1$.
\end{itemize}
 Thus the estimate \eqref{ord(wkn)=-k+n-1} has been proved for all $k$
 and $n$.

 This shows that $w(1) = \sum_{k=0}^\infty w_k = \sum_{k=0}^\infty
 \sum_{n=0}^\infty \hbar^n w_{k,n}$ makes sense as a formal series of
 $\hbar$, $\xi$ and $\xi'$. Moreover it is obvious that $w_k$ and $w(1)$
 are formally regular with respect to $\hbar$.

 As a next step, we expand $\psi(t)$ as $\psi(t) = \sum_{k=0}^\infty
 \psi_k t^k$ and rewrite \eqref{diff-eq:psi} into a recursion relation
 and the initial condition:
\begin{equation}
 \begin{split}
    \psi_{k+1} &= \frac{1}{k+1}
    \left(
     \hbar \xi\der_\xi \der_{s'} \psi_k +
     \sum_{\nu=0}^k (\xi\der_\xi \psi_\nu \der_{s'} w_{k-\nu}
                   + \xi\der_\xi w_{k-\nu} \der_{s'} \psi_\nu)
    \right),
\\
    \psi_0 &= a(s,\xi) b(s',\xi')
 \end{split}
\label{recursion:psi}
\end{equation}
 To prove the convergence as a formal power series, we expand $\psi_k$
 as 
\begin{equation}
    \psi_k(\hbar,s,s',\xi,\xi') 
    =
    \sum_{n=0}^\infty \hbar^n \psi_{k,n}(s,s',\xi,\xi'),
\label{psik:h-expansion}
\end{equation}
 and rewrite the recursion relation as follows.
\begin{equation}
 \begin{split}
    \psi_{k+1,n} &= \frac{1}{k+1}
    \left(
     \xi\der_\xi \der_{s'} \psi_{k,n-1} +
     \sum_{\substack{k'+k''=k\\n'+n''=n}}
       (\xi\der_\xi \psi_{k',n'} \cdot \der_{s'} w_{k'',n''}
     + \der_{s'} \psi_{k',n''} \cdot \xi\der_\xi w_{k'',n''})
    \right),
\\
    \psi_0 &= a(s,\xi) b(s',\xi')
 \end{split}
\label{recursion:psi(kn)}
\end{equation}
 Our assumption being $\ordx a(s,\xi)=M$ and $\ordx b(s',\xi')=N$, we
 have 
\begin{equation}
     \ordx \psi_{0,n} \leqq M+N.
\label{ordx(psi0n)=M+N}
\end{equation}
 As in the estimate of $w_{k,n}$, we prove 
\begin{equation}
    \ordx \psi_{k,n} \leqq \min(M+N, M+N-k+n)
\label{ord(psikn)=M+N-k+n}
\end{equation}
 by induction. 
\begin{itemize}
 \item First, \eqref{ord(psikn)=M+N-k+n} holds for $k=0$ and any
       $n\geqq0$ because of \eqref{ordx(psi0n)=M+N}.
 \item Assume that \eqref{ord(psikn)=M+N-k+n} holds for $k=0,\dots,k_0$
       and $n\geqq 0$. The right hand side of \eqref{recursion:psi(kn)}
       with $k=k_0$ has the order with respect to $\xi$ and $\xi'$ not
       more than $M+N-k_0+(n-1)=(M+N-k'+n')+(-k''+n''-1)=M+N-k_0+n-1$
       nor $M+N$ because of the induction hypothesis and
       \eqref{ord(wkn)=-k+n-1}.
 \item This proves \eqref{ord(psikn)=M+N-k+n} for $k=k_0+1$.
\end{itemize}
 Thus we have proved \eqref{ord(psikn)=M+N-k+n} for any $k$ and $n$,
 which shows that the inifinite sum $\psi(1) = \sum_{k=0}^\infty \psi_k
 = \sum_{k=0}^\infty \sum_{n=0}^\infty \hbar^n \psi_{k,n}$
 makes sense. The regularity of $\psi_k$ and $\psi(1)$ is also obvious.

 We have constructed $\pi(t) = \pi(t;\hbar,s,s',\xi,\xi') =
 \psi(t;\hbar,s,s',\xi,\xi') e^{w(t;\hbar,s,s',\xi,\xi')}$, which is
 meaningful also at $t=1$. Hence the product $a(\hbar,s,\xi) \circ
 b(\hbar,s,\xi) = \pi(1; \hbar,s,s,\xi,\xi)$ is expressed in the form
 $c(\hbar,s,\xi) e^{r(\hbar,s,\xi)/\hbar}$, where
 $c(\hbar,s,\xi)=\psi(1;\hbar,s,s,\xi,\xi)$, $r(\hbar,s,\xi)=
 w(1;\hbar,s,s,\xi,\xi)$.
\end{proof}

\begin{proof}[Proof of \propref{prop:exp(:X:)=:exp(S):}]
 We make use of differential equations satisfied by the operator
\begin{equation}
    E(t)=E(t;\hbar,s,e^{\hbar\ders}) :=
    \exp\left(\frac{t}{\hbar}X(\hbar,s,e^{\hbar\ders})\right),
\label{def:E(t)}
\end{equation}
 depending on a parameter\footnote{Of course this parameter $t$ does not
 have any relation with the time variables of the Toda lattice
 hierarchy. It is not the same $t$ in the proof of
 \lemref{lem:prod-exp-symbol}, either.} $t$. The total symbol of $E(t)$ is
 defined as
\begin{equation}
    E(t;\hbar,s,\xi) 
    = \sum_{k=0}^\infty \frac{t^k}{\hbar^k k!} X^{(k)}(\hbar,s,\xi),
    \qquad
    X^{(0)}=1, \qquad X^{(k+1)} = X \circ X^{(k)}.
\label{E:symbol}
\end{equation}
 Taking the logarithm (as a function, not as an operator) of this, we
 can define $S(t)=S(t;\hbar,s,\xi)$ by
\begin{equation}
    E(t;\hbar,s,\xi) =
    e^{\hbar^{-1} S(t;\hbar,s,\xi)}
\label{E=eS}
\end{equation}
 What we are to prove is that $S(t)$, constructed as a series, makes
 sense at $t=1$ and formally regular with respect to $\hbar$.

 Differentiating \eqref{E=eS}, we have
\begin{equation}
    X(\hbar,s,\xi) \circ E(t;\hbar,s,\xi)
    =
    \frac{\der S}{\der t} e^{S(t;\hbar,s,\xi)/\hbar}
\label{XE=dS/dtE}
\end{equation}
 By \lemref{lem:prod-exp-symbol} ($a \mapsto X$, $b \mapsto 1$, $p
 \mapsto 0$, $q \mapsto S$) and the technique in its proof, we can
 rewrite the left hand side as follows. (Hereafter we sometimes omit the
 argument $\hbar$ of functions for brevity.): 
\begin{equation}
    X(s,\xi) \circ E(t;s,\xi)
    =
    Y(t;s,s,\xi,\xi) e^{S(t;s,\xi)/\hbar}
\label{XE=Y*eq}
\end{equation}
 where $Y(t;s,s',\xi,\xi') = \sum_{k=0}^\infty Y_k$ and
 $Y_k(t;s,s',\xi,\xi')$ are defined by
\begin{equation}
 \begin{split}
    Y_{k+1}(t;s,s',\xi,\xi') &=
    \frac{1}{k+1} (\hbar \xi\der_\xi \der_{s'} Y_k(t;s,s',\xi,\xi')
    + \xi\der_\xi Y_k(t;s,s',\xi,\xi') \der_{s'} S(t;s',\xi')),
\\
    Y_0 (t;s,s',\xi,\xi') &= X(s,\xi).
 \end{split}
\label{recursion:Y:1}
\end{equation}
 $Y_k(t)$ corresponds to $\psi_k$ in the proof of
 \lemref{lem:prod-exp-symbol}, while $w_k$ there is $\delta_{k,0}
 S(t)$. (Recall that the role of $t$ is different. The parameter $t$ in
 the proof of \lemref{lem:prod-exp-symbol} is set to $1$ here.) On the
 other hand, substituting \eqref{XE=Y*eq} into the left hand side of
 \eqref{XE=dS/dtE}, we have
\begin{equation}
    \frac{\der S}{\der t}(t;s,\xi)
    =
    Y(t;s,s,\xi,\xi)
\label{dS/dt=Y}
\end{equation}

 We rewrite the system \eqref{recursion:Y:1} and \eqref{dS/dt=Y} in
 terms of expansion of $S(t;s,\xi)=S(t;\hbar,s,\xi)$ and
 $Y_k(t;s,s',\xi,\xi')=Y_k(t;\hbar,s,s',\xi,\xi')$ in powers of $t$ and
 $\hbar$: 
\begin{equation}
 \begin{split}
    S(t;\hbar,s,\xi) &= \sum_{l=0}^\infty S^{(l)}(\hbar,s,\xi) t^l
    = \sum_{l=0}^\infty \sum_{n=0}^\infty S^{(l)}_n (s,\xi) \hbar^n t^l,
\\
    Y_k(t;\hbar,s,s',\xi,\xi') 
    &= \sum_{l=0}^\infty Y^{(l)}_k(\hbar,s,s',\xi,\xi') t^l
    = \sum_{l=0}^\infty \sum_{n=0}^\infty 
    Y^{(l)}_{k,n} (s,s',\xi,\xi') \hbar^n t^l,
 \end{split}
\label{expansion:S,Y}
\end{equation}
 The coefficient of $\hbar^n t^l$ in the recursion relation
 \eqref{recursion:Y:1} is
\begin{multline}
    Y^{(l)}_{k+1,n}(s,s',\xi,\xi')
\\    =
    \frac{1}{k+1}
    \left(
    \xi\der_\xi \der_{s'} Y^{(l)}_{k,n-1}(s,s',\xi,\xi')
    +
    \sum_{\substack{l'+l''=l \\ n'+n''=n}}
    \xi\der_\xi Y^{(l')}_{k,n'}(s,s',\xi,\xi')
    \der_{s'} S^{(l'')}_{n''}(s',\xi')
    \right)
\label{recursion:Y(l,k,n)}
\end{multline}
 ($Y^{(l)}_{k,-1}=0$) while \eqref{dS/dt=Y} gives
\begin{equation}
    S^{(l+1)}_n(s,\xi)
    = \frac{1}{l+1} \sum_{k=0}^\infty Y^{(l)}_{k,n}(s,s,\xi,\xi)
\label{recursion:S(l+1,n)}
\end{equation}
 We first show that these recursion relations consistently determine
 $Y^{(l)}_{k,n}$ and $S^{(l)}_n$. Then we prove that the infinite sum
 in \eqref{recursion:S(l+1,n)} is finite.

 Fix $n\geqq 0$ and assume that $Y^{(l)}_{k,0},\dotsc,Y^{(l)}_{k,n-1}$
 and $S^{(l)}_0,\dotsc,S^{(l)}_{n-1}$ have been determined for all
 $(l,k)$. (When $n=0$, $Y^{(l)}_{k,-1}=0$ as mentioned above and
 $S^{(l)}_{-1}$ can be ignored as it does not appear in the induction.)

\begin{enumerate}
 \item Since $E(t=0)=1$ by the definition \eqref{def:E(t)}, we have
       $S^{(0)}=0$. Hence
\begin{equation}
    S^{(0)}_n=0.
\label{S0n=0}
\end{equation}

 \item From the initial condition in \eqref{recursion:Y:1} we have
\begin{equation}
    Y^{(l)}_{0,n}(s,s',\xi,\xi') = \delta_{l,0} X_n(s,\xi).
\label{Yl0n}
\end{equation}
       It follows from this equation and the assumption
       \eqref{X:alpha=0} that
\begin{equation}
    \ordx Y^{(0)}_{0,n} \leqq -1.
\label{ordY00n<=-1}
\end{equation}

\item When $l=0$, the second sum in the right hand side of the recursion
      relation \eqref{recursion:Y(l,k,n)} is absent because of
      \eqref{S0n=0}. Hence if $n\geqq k+1$, we have
\[
    Y^{(0)}_{k+1,n} = \frac{1}{k+1}\xi\der_\xi\der_{s'} Y^{(0)}_{k,n-1}
    = \dots 
    = \frac{1}{(k+1)!}(\xi\der_\xi\der_{s'})^{k+1} Y^{(0)}_{0,n-k-1} = 0 
\]
      since $Y^{(0)}_{0,n-k-1}$  does not depend on $s'$ thanks to
      \eqref{Yl0n}. If $n<k+1$, the above expression becomes zero
      by $Y^{(0)}_{k-n+1,-1}=0$. Hence together with \eqref{Yl0n}, we
      obtain 
\begin{equation}
    Y^{(0)}_{k,n}=\delta_{k,0} X_n.
\label{Y0kn}
\end{equation}

 \item By \eqref{recursion:S(l+1,n)} we can determine $S^{(1)}_n$:
\begin{equation}
    S^{(1)}_n = \sum_{k=0}^\infty Y^{(0)}_{k,n} = Y^{(0)}_{0,n} = X_n.
\label{S(1,n)=Xn}
\end{equation}
       In particular, 
\begin{equation}
    \ordx \der_{s'} S^{(1)}_n = \ordx \der_{s'} X_n \leqq -1.
\label{ordx(dS1n)<=-1}
\end{equation}

 \item Fix $l_0\geqq 1$ and assume that for all $l=0,\dotsc,l_0-1$ and
       for all $k=0,1,2,\dotsc$, we have determined $Y^{(l)}_{k,n}$ and
       that for all $l=0,\dots,l_0$ we have determined
       $S^{(l)}_{n}$. (The steps (3) and (4) are for $l_0=1$.)

       Since $S^{(0)}_{n''}=0$ by \eqref{S0n=0}, the index $l'$ in the
       right hand side of the recursion relation
       \eqref{recursion:Y(l,k,n)} (with $l=l_0$) runs essentially from
       $0$ to $l_0-1$. Hence this relation determines
       $Y^{(l_0)}_{k+1,n}$ from known quantities for all $k\geqq 0$.

       Because of the initial condition $Y_0(t;s,s',\xi,\xi')=X(s,\xi)$
       (cf.\ \eqref{recursion:Y:1}) $Y_0$ does not depend on $t$, which
       means that its Taylor coefficients $Y^{(l_0)}_{0,n}$ vanish for
       all $l_0\geqq 1$:
\begin{equation}
    Y^{(l_0)}_{0,n} = 0.
\end{equation}
       Thus we have determined all $Y^{(l_0)}_{k,n}$ ($k=0,1,2,\dotsc$).
 
 \item We shall prove below that $Y^{(l_0+1)}_{k,n}=0$ if $k>l_0+n+1$.
       Hence the sum in \eqref{recursion:S(l+1,n)} is finite and
       $S^{(l_0+1)}_n$ is determined. The induction proceeds by
       incrementing $l_0$ by one.
\end{enumerate}
 In this way induction proceeds and all $Y^{(l)}_{k,n}$ and $S^{(l)}_n$
 are determined.

\medskip
 Let us prove that $Y^{(l)}_{k,n}$'s determined above satisfy
\begin{gather}
    Y^{(l)}_{k,n}=0, \qquad\text{if\ }k>l+n,
\label{Ylkn=0}
\\
    \ordx Y^{(l)}_{k,n} \leqq -l-1, \qquad\text{if\ }0\leqq k \leqq l+n,
\label{ordYlkn<=-l-1}
\end{gather}
 (We define that $\ordx 0 = -\infty$.) In particular, the sum in
 \eqref{recursion:S(l+1,n)} is well-defined and
\begin{equation}
    \ordx S^{(l+1)}_n \leqq -l-1.
\label{ordS(l+1)<-l-1}
\end{equation} 

 If $n=-1$, both \eqref{Ylkn=0} and \eqref{ordYlkn<=-l-1} are
 obvious. Fix $n_0\geqq 0$ and assume that we have proved \eqref{Ylkn=0}
 and \eqref{ordYlkn<=-l-1} for $n<n_0$ and all $(l,k)$.

 When $n=n_0$ and $l=0$, \eqref{Ylkn=0} and \eqref{ordYlkn<=-l-1} are
 true for all $k$ because of \eqref{Y0kn} and \eqref{ordY00n<=-1}. 

 Fix $l_0\geqq 0$ and assume that we have proved \eqref{Ylkn=0} and
 \eqref{ordYlkn<=-l-1} for $n=n_0$, $l\leqq l_0$ and all $k$. As a
 result \eqref{ordS(l+1)<-l-1} is true for $l\leqq l_0$.

 For $(n,l,k)=(n_0,l_0+1,0)$ \eqref{Ylkn=0} is void and
 \eqref{ordYlkn<=-l-1} is true because of \eqref{Yl0n} and $\ordx
 X_{n_0}(s,\xi)\leqq -1$.

 Put $n=n_0$ and $l=l_0+1$ in \eqref{recursion:Y(l,k,n)} and assume that
 $k+1>(l_0+1)+n_0$. Then $k>(l_0+1)+(n_0-1)$, which guarantees that
 $Y^{(l)}_{k,n-1}=Y^{(l_0+1)}_{k,n_0-1}=0$ by the induction hypothesis
 on $n$. As we mentioned in the step (5) above, $l'$ in the right hand
 side of \eqref{recursion:Y(l,k,n)} runs from $0$ to $l-1=l_0$. Hence,
 as we are assuming that $k>l_0+n$, we have $k>l'+n'$, which leads to
 $Y^{(l')}_{k,n'}=0$ by the induction hypothesis on $l$ and
 $n$. Therefore all terms in the right hand side of
 \eqref{recursion:Y(l,k,n)} vanish and we have
 $Y^{(l_0+1)}_{k+1,n_0}=0$. The induction on $k$ for \eqref{Ylkn=0} is
 completed, namely it is proved for $n=n_0$, $l=l_0+1$ and $k\geqq 1$.

 The estimate \eqref{ordYlkn<=-l-1} is easy to check for $n=n_0$,
 $l=l_0+1$ and $k\geqq 1$ by the recursion relation
 \eqref{recursion:Y(l,k,n)}. (Recall once again that $\xi\der_\xi$ does
 not change the order.)

 The step $l=l_0+1$ being proved, the induction proceeds with respect to
 $l$ and consequently with respect to $n$.

\bigskip
 In summary we have constructed $Y(t;s,s',\xi,\xi')$ and $S(t;s,\xi)$
 satisfying \eqref{XE=Y*eq} and \eqref{dS/dt=Y}. Moreover, thanks to
 \eqref{Ylkn=0} the sum over $k$ for each fixed $(n,l)$ in 
\begin{equation}
    Y(1;s,s',\xi,\xi') =
    \sum_{n=0}^\infty \sum_{l=0}^\infty \sum_{k=0}^\infty 
    Y^{(l)}_{k,n}(s,s',\xi,\xi') \hbar^n,
\end{equation}
 is finite and the sum over $l$ is meaningful as a power series of $\xi$
 because of \eqref{ordYlkn<=-l-1}. The series 
\begin{equation}
    S(1;s,\xi) =
    \sum_{n=0}^\infty \sum_{l=0}^\infty S^{(l)}_n(s,\xi)\hbar^n,
\end{equation}
 is also meaningful as a power series of $\xi$ thanks to
 \eqref{ordS(l+1)<-l-1}.

 Thus \propref{prop:exp(:X:)=:exp(S):} is proved.
\end{proof}

\begin{proof}[Proof of \propref{prop::exp(S):=exp(:X:)}]
 We reverse the order of the previous proof. Namely, given
 $S(\hbar,s,\xi)$, we shall construct $X(\hbar,s,\xi)$ such that the
 corresponding $S(1;\hbar,s,\xi)$ in the above proof coincides with it.

 Suppose we have such $X(\hbar,s,\xi)$. Then the above procedure
 determine $Y^{(l)}_{k,n}$ and $S^{(l)}_n$. We expand them as follows:
\begin{align*}
    S(\hbar,s,\xi) &= \sum_{n=0}^\infty S_n(s,\xi) \hbar^n 
    = \sum_{n=0}^\infty \sum_{j=1}^\infty S_{n,j}(s,\xi) \hbar^n,
\\
    X(\hbar,s,\xi) &= \sum_{n=0}^\infty X_n(s,\xi) \hbar^n 
    = \sum_{n=0}^\infty \sum_{j=1}^\infty X_{n,j}(s,\xi) \hbar^n,
\\
    S(t;\hbar,s,\xi)& 
    = \sum_{l=0}^\infty \sum_{n=0}^\infty S^{(l)}_n(s,\xi) \hbar^n t^l
    = \sum_{l=0}^\infty \sum_{n=0}^\infty \sum_{j=1}^\infty 
    S^{(l)}_{n,j}(s,\xi) \hbar^n t^l,
\\
    Y_k(t;\hbar,s,s',\xi,\xi')
    &= \sum_{l=0}^\infty  \sum_{n=0}^\infty 
    Y^{(l)}_{k,n}(s,s',\xi,\xi') \hbar^n t^l
\\
    &= \sum_{l=0}^\infty \sum_{n=0}^\infty \sum_{j=1}^\infty 
    Y^{(l)}_{k,n,j}(s,s',\xi,\xi') \hbar^n t^l
\end{align*}
 Here terms with index $j$ are homogeneous terms of degree $-j$ with
 respect to $\xi$ and $\xi'$.

 At the end of this proof we shall determine $X_n$ by \eqref{Yl0n},
\begin{equation}
    X_n(s,\xi) = Y^{(0)}_{0,n}(s,s',\xi,\xi').
\label{Xn=Y(0,0,n)}
\end{equation}
 (In particular, $Y^{(0)}_{0,n}(s,s',\xi,\xi')$ should not depend on
 $s'$ and $\xi'$.) For this purpose, $Y^{(0)}_{0,n}$ should be
 determined by
\begin{equation}
    Y^{(0)}_{0,n}(s,s',\xi,\xi')
    =
    S_n(s,\xi) 
    - \sum_{\substack{(l,k)\neq(0,0)\\l,k\geq 0}}
      \frac{1}{l+1} Y^{(l)}_{k,n}(s,s,\xi,\xi)
\label{Y00n=Sn-Ylkn}
\end{equation}
 because of \eqref{recursion:S(l+1,n)} and
 $S_n(s,\xi)=S_n(t=1;s,\xi)=\sum_{l=0}^\infty S^{(l)}_n(s,\xi)$.

 Since $\ordx Y^{(l)}_{k,n}$ should be not more than $-l-1$ (cf.\
 \eqref{ordYlkn<=-l-1}), we expect $Y^{(l)}_{k,n,1}=0$ for $l>0$. For
 $l=0$ and $k>0$ $Y^{(0)}_{k,n,1}= 0$ follows from \eqref{Y0kn}. Hence
 picking up homogeneous terms of degree $-1$ with respect to $\xi$ from
 \eqref{Y00n=Sn-Ylkn}, the following equation should hold:
\begin{equation}
    Y^{(l)}_{k,n,1}=\delta_{l,0} \delta_{k,0} S_{n,1}
\label{Y00n1=Sn1}
\end{equation}
 All $Y^{(l)}_{k,n,1}$ are determined by this condition. Note also that
\begin{equation}
    Y^{(l)}_{0,n,j}=0 \ \text{for\ }l\neq 0
\label{Yl0nj=0}
\end{equation}
 because $Y_0$ should not depend on $s$ because of \eqref{Yl0n}.

 Having determined initial conditions in this way, we shall determine
 $Y^{(l)}_{k,n,j}$ inductively. To this end we rewrite the recursion
 relation \eqref{recursion:Y(l,k,n)} by \eqref{recursion:S(l+1,n)} and
 pick up homogeneous terms of degree $-j$:
\begin{multline}
    Y^{(l)}_{k+1,n,j}(s,s',\xi,\xi')
    =
    \frac{1}{k+1}
    \Biggl(
    \xi\der_\xi \der_{s'} Y^{(l)}_{k,n-1,j}(s,s',\xi,\xi')
\\
    +
    \sum_{\substack{l'+l''=l,\ l''\geqq 1,\ 
                    j'+j''=j,\ j',j''\geqq 1\\ 
                    n'+n''=n,\ 0\leqq k''}}
    \frac{1}{l''}
    \xi\der_\xi Y^{(l')}_{k,n',j'}(s,s',\xi,\xi')
    \der_{s'} Y^{(l''-1)}_{k'',n'',j''}(s',s',\xi',\xi')
    \Biggr)
\label{recursion:Y(l,k,n)<-Y}
\end{multline}
 (As before, terms like $Y^{(l)}_{k,-1,j-1}$ appearing the above
 equation for $n=0$ can be ignored.)

 Fix $n_0\geqq 0$ and assume that $Y^{(l)}_{k,0,j}, \dotsc,
 Y^{(l)}_{k,n_0-1,j}$ are determined for all $(l,k,j)$.
\begin{enumerate}
 \item First we determine $Y^{(l)}_{k,n_0,1}$ for all $(l,k)$ by
       \eqref{Y00n1=Sn1}. (This is consistent with the recursion
       relation \eqref{recursion:Y(l,k,n)<-Y}.)

 \item Fix $j_0\geqq 2$ and assume that $Y^{(l)}_{k,n_0,j}$ are
       determined for $j=1,\dots,j_0-1$ and all $(l,k)$. (The above
       step is for $j_0=2$.)

       Since all the quantities in the right hand side of the recursion
       relation \eqref{recursion:Y(l,k,n)<-Y} with $j=j_0$ are known by
       the induction hypothesis, we can determine
       $Y^{(l)}_{k,n_0,j_0}$ for $l=0,1,2,\dots$ and $k=1,2,\dots$.

 \item Together with \eqref{Yl0nj=0}, $Y^{(l)}_{0,n_0,j_0}=0$
       for $l=1,2,\dots$, we have determined all $Y^{(l)}_{k,n_0,j_0}$
       except for the case $(l,k)=(0,0)$.

 \item It follows from \eqref{recursion:Y(l,k,n)<-Y} and
       \eqref{Y00n1=Sn1} by induction that all
       $Y^{(l)}_{k,n_0,j}$ determined in (1), (2) and (3) satisfy the
       following properties:
\begin{itemize}
 \item if $k>l+n$, then 
\begin{equation}
    Y^{(l)}_{k,n,j} = 0,
\label{Ylknj=0:k>l+n}
\end{equation}
       which corresponds to \eqref{Ylkn=0} in the proof of
       \propref{prop:exp(:X:)=:exp(S):}; 
 \item if $0\leqq k \leqq l+n$ and $j \leqq l$, then 
\begin{equation}
    Y^{(l)}_{k,n,j} = 0,
\label{Ylknj=0:j<=l}
\end{equation}
       which corresponds to \eqref{ordYlkn<=-l-1} in the proof of
       \propref{prop:exp(:X:)=:exp(S):}.
\end{itemize}

 \item We determine $Y^{(0)}_{0,n_0,j_0}$ by
\begin{equation}
    Y^{(0)}_{0,n_0,j_0}
    =
    S_{n_0,j_0}
    - \sum_{\substack{(l,k)\neq(0,0)\\l,k\geq 0}}
      \frac{1}{l+1} Y^{(l)}_{k,n_0,j_0}(s,s,\xi,\xi)
\label{Y00nj=Snj-Ylknj}
\end{equation}
       which is the homogeneous part of degree $-j_0$ in
       \eqref{Y00n=Sn-Ylkn}. The sum in the right hand side is finite
       because of \eqref{Ylknj=0:k>l+n} and \eqref{Ylknj=0:j<=l}.

 \item The induction with respect to $j$ proceeds by incrementing $j_0$.
\end{enumerate}
 Thus all $Y^{(l)}_{k,n_0,j}$ are determined and $X_{n_0}$ is determined
 by \eqref{Xn=Y(0,0,n)}, namely, $X_{n_0}(x,\xi)=\sum_{j=1}^\infty
 Y^{(0)}_{0,n_0,j}$ (cf.\ \eqref{Xn=Y(0,0,n)}), which completes the
 proof of \propref{prop::exp(S):=exp(:X:)}.
\end{proof}

\section{Asymptotics of the tau function}
\label{sec:tau-function}

In this section we derive an $\hbar$-expansion 
\begin{equation}
    \log\tau(\hbar,t,\tbar,s) = \sum_{n=0}^\infty \hbar^{n-2} F_n(t,\tbar,s)
\label{tau}
\end{equation}
of the tau function (cf.\ \eqref{tau/tau}) from the $\hbar$-expansion
of the $S$-functions $S(\hbar,t,\tbar,s;z)$ \eqref{S} and
$\Sbar(\hbar,t,\tbar,s;\zbar)$ \eqref{Sbar}.

Let us recall the fundamental relations \eqref{tau/tau} between the wave
functions and the tau function again:
\begin{equation}
 \begin{aligned}
    \Psi(\hbar,t,\tbar;z) &=
    \frac{ \tau(\hbar,t-\hbar [z^{-1}] ,\tbar,s) }
         { \tau(\hbar,t,\tbar,s) }
    z^{s/\hbar} e^{\zeta(t,z)/\hbar},
\\
    \Psibar(\hbar,t,\tbar;\zbar) &=
    \frac{ \tau(\hbar,t,\tbar-\hbar [\zbar] ,s+\hbar) }
         { \tau(\hbar,t,\tbar,s) }
    \zbar^{s/\hbar} 
    e^{\zeta(\tbar,\zbar^{-1})/\hbar}
 \end{aligned}
\end{equation}
where $[z^{-1}] =(1/z,1/2z^2, 1/3z^3,\dots)$,
$\zeta(t,z)=\sum_{n=1}^\infty t_n z^n$ etc. (Here we again omit
inessential constants, $\alpha(\hbar)$ and $\alphabar(\hbar)$.) This
implies that
\begin{align}
    \hbar^{-1} \hat S(\hbar,t,\tbar,s;z)
    &=
    \left(
    e^{-\hbar D(z)} - 1
    \right) \log\tau(\hbar,t,\tbar,s),
\label{tau->Shat}
\\
    \hbar^{-1} \hat \Sbar(\hbar,t,\tbar,s;\zbar)
    &=
    \left(
    e^{-\hbar \Dbar(\zbar)} e^{\hbar\ders} - 1
    \right) \log\tau(\hbar,t,\tbar,s),
\label{tau->Shatbar}
\end{align}
where
\begin{equation}
 \begin{aligned}
    \hat S(\hbar,t,\tbar,s;z) &= S(\hbar,t,\tbar,s;z)- \zeta(t,z),
\\    
    \hat \Sbar(\hbar,t,\tbar,s;\zbar)
    &=
    \Sbar(\hbar,t,\tbar,s;\zbar) - \zeta(\tbar,\zbar^{-1}),
 \end{aligned}
\label{def:Shat,Shatbar}
\end{equation}
and 
\begin{equation}
    D(z) =\sum_{j=1}^\infty \frac{z^{-j}}{j} \frac{\der}{\der t_j},
    \qquad
    \Dbar(\zbar) =\sum_{j=1}^\infty 
              \frac{\zbar^j}{j} \frac{\der}{\der \tbar_j}.
\end{equation}
Differentiating \eqref{tau->Shat} with respect to $z$, we have
\begin{equation}
 \begin{split}
    \hbar^{-1} \frac{\der}{\der z} \hat S(\hbar,t,\tbar,s;z)
    =&
    - \hbar D'(z) e^{-\hbar D(z)} \log\tau(\hbar,t,\tbar,s)
\\
    =&
    - \hbar D'(z) (\hbar^{-1} \hat S(\hbar,t,\tbar,s;z) 
    + \log\tau(\hbar,t,\tbar,s)),
 \end{split}
\label{(tau->Shat)'}
\end{equation}
where $D'(z):=\frac{\der}{\der z} D(z) = - \sum_{j=1}^\infty z^{-j-1}
\frac{\der}{\der t_j}$. Hence
\begin{equation}
    -\hbar D'(z) \log\tau(\hbar,t,\tbar,s)
    =
    \hbar^{-1} \left(\frac{\der}{\der z} + \hbar D'(z)\right)
    \hat S(\hbar,t,\tbar,s;z)
\label{Shat->tau}
\end{equation}
Multiplying $z^n$ to this equation and taking the residue, we obtain a
system of differential equations
\begin{equation}
    \hbar\frac{\der}{\der t_n} \log \tau(\hbar,t,\tbar,s)
    =
    \hbar^{-1} \Res_{z=\infty} z^n 
    \left(\frac{\der}{\der z} + \hbar D'(z)\right) 
     \hat S(\hbar,t,\tbar,s;z)\, dz
\label{dtau/dtn}
\end{equation}
for $n=1,2,\dotsc$. In the same way we have
\begin{equation}
    -\hbar \Dbar'(\zbar) \log\tau(\hbar,t,\tbar,s)
    =
    \hbar^{-1} \left(\frac{\der}{\der \zbar} + \hbar \Dbar'(\zbar)\right)
    \hat \Sbar(\hbar,t,\tbar,s;\zbar)
\label{Shatbar->tau}
\end{equation}
and
\begin{equation}
    \hbar\frac{\der}{\der \tbar_n} \log \tau(\hbar,t,\tbar,s)
    =
    -\hbar^{-1} \Res_{\zbar=0} \zbar^{-n} 
    \left(\frac{\der}{\der \zbar} + \hbar \Dbar'(\zbar)\right)
     \hat \Sbar(\hbar,t,\tbar,s;\zbar)\, d\zbar,
\label{dtau/dtbarn}
\end{equation}
for $n=1,2,\dotsc$, from \eqref{tau->Shatbar}. By putting $\zbar=0$ in
\eqref{tau->Shatbar} we have a difference equation for the tau function:
\begin{equation*}
    (e^{\hbar\ders} - 1) \log\tau(\hbar,t,\tbar,s)
    =
    \hbar^{-1} \hat\Sbar(\hbar,t,\tbar,s;0).
\end{equation*}
In fact, it follows from \eqref{Wbar=exp(phi)exp(Xbar)},
\eqref{def:wave-func} and \eqref{Sbar} that
$\hat\Sbar(\hbar,t,\tbar,s;0)=\phi(\hbar,t,\tbar)$. Hence we have
\begin{equation}
    \hbar (e^{\hbar\ders} - 1) \log\tau(\hbar,t,\tbar,s)
    =
    \phi(\hbar,t,\tbar,s).
\label{tau(s+hbar)}
\end{equation}
As is shown in \cite{uen-tak:84}, the system \eqref{dtau/dtn},
\eqref{dtau/dtbarn} and \eqref{tau(s+hbar)} is compatible and determines
the tau function up to multiplicative constant.

By substituting the $\hbar$-expansions
\begin{align}
    \log\tau(\hbar,t,\tbar,s) 
    &= \sum_{n\in\Integer} \hbar^{n-2} F_n(t,\tbar,s),
\label{logtau}
\\
    \hat S(\hbar,t,\tbar,s;z)
    &= \sum_{n=0}^\infty \hbar^n S_n(t,\tbar,s;z),
\label{logS}
\\
    \hat \Sbar(\hbar,t,\tbar,s;\zbar)
    &= \sum_{n=0}^\infty \hbar^n \Sbar_n(t,\tbar,s;\zbar)
\label{logSbar}
\end{align}
and \eqref{phi:h-expansion} into \eqref{Shat->tau}, \eqref{Shatbar->tau}
and \eqref{tau(s+hbar)}, we have
\begin{align}
    \sum_{j=1}^\infty \sum_{n\in\Integer} z^{-j-1} \hbar^{n-1}
    \frac{\der F_n}{\der t_j}
    &=
    \sum_{n=0}^\infty \left(
    \hbar^{n-1} \frac{\der S_n}{\der z}
    - 
    \sum_{j=1}^\infty z^{-j-1} \hbar^n \frac{\der S_n}{\der t_j}
    \right).
\label{S->F}
\\
    -\sum_{j=1}^\infty \sum_{n\in\Integer} \zbar^{j-1} \hbar^{n-1}
    \frac{\der F_n}{\der \tbar_j}
    &=
    \sum_{n=0}^\infty \left(
    \hbar^{n-1} \frac{\der \Sbar_n}{\der \zbar}
    +
    \sum_{j=1}^\infty \zbar^{j-1} \hbar^n \frac{\der \Sbar_n}{\der \tbar_j}
    \right).
\label{Sbar->F}
\\
    \sum_{n\in\Integer} \hbar^{n-1} \left(
     \sum_{m=1}^n \frac{1}{m!} \frac{\der^m F_{n-m}}{\der s^m}
    \right)
    &=
    \sum_{n=0}^\infty \hbar^n \phi_n.
\label{phi->F}
\end{align}
It is obvious from these equations that $F_n= \text{const.}$ for
$n<0$. Therefore we can conclude that $\log\tau$ has the expansion
\eqref{tau}. 

Let us expand $S_n(t;z)$ and $\Sbar_n(t;\zbar)$ into a power series of
$z^{-1}$ and $\zbar$: 
\begin{equation}
    S_n(t;z) = - \sum_{k=1}^\infty \frac{z^{-k}}{k} v_{n,k},
    \qquad
    \Sbar_n(t;\zbar) 
    = \phi_n+\sum_{k=1}^\infty \frac{\zbar^k}{k} \vbar_{n,k}.
\label{Sn,Sbar:expand} 
\end{equation}
(The notation is chosen so that it is consistent with our previous work,
e.g., \cite{tak-tak:95}.) Comparing the coefficients of $z^{-j-1}
\hbar^{n-1}$ in \eqref{S->F} and the coefficients of $\zbar^{j-1}
\hbar^{n-1}$ in \eqref{Sbar->F}, we have the equations
\begin{align}
    \frac{\der F_n}{\der t_j}
    &=
    v_{n,j} 
    + \sum_{\substack{k+l=j\\ k\geq 1, l\geq 1}}
      \frac{1}{l}\, \frac{\der v_{n-1,l}}{\der t_k}
    \qquad (v_{-1,j}=0), 
\label{dFn/dtj}
\\
    -\frac{\der F_n}{\der \tbar_j}
    &=
    \vbar_{n,j} + \frac{\der\phi_n}{\der \tbar_j}
    + \sum_{\substack{k+l=j\\ k\geq 1, l\geq 1}}
      \frac{1}{l}\, \frac{\der \vbar_{n-1,l}}{\der \tbar_k}
    \qquad (\vbar_{-1,j}=0),
\label{dFn/dtbarj}
\end{align}
for $n=0,1,2,\dotsc$.

{}From the equation \eqref{phi->F} it is easy to see that $\der F_n/\der
s$ is determined recursively. Let us rewrite it in more explicit
way. First arrange the coefficients of $\hbar^{n-1}$ in \eqref{phi->F}
in the vector form:
\begin{equation}
    \begin{pmatrix}
     \ders &  &  &  &
\\ \\
     \frac{1}{2!}\ders^2 & \ders & & &
\\ \\
     \frac{1}{3!}\ders^3 & \frac{1}{2!}\ders^2 & \ders & &
\\
     & \ddots & \ddots & \ddots & 
    \end{pmatrix}
    \begin{pmatrix} F_0 \\ F_1 \\ F_2 \\ \vdots \end{pmatrix}
    =
    \begin{pmatrix} \phi_0 \\ \phi_1 \\ \phi_2 \\ \vdots \end{pmatrix}.
\label{phi->F:matrix}
\end{equation}
The matrix in the left hand side is 
\begin{equation*}
    \sum_{n=0}^\infty \frac{\ders^{n+1}}{(n+1)!} \Lambda^{-n}
    =
    \left.\frac{e^T-1}{T}\right|_{T=\ders \Lambda^{-1}} \ders
\end{equation*}
where $\Lambda^{-n}$ is the shift matrix
$(\delta_{i-n,j})_{i,j=1}^\infty$. Hence, applying the matrix
\begin{equation*}
    \left.\frac{T}{e^T-1}\right|_{T=\ders \Lambda^{-1}}
\end{equation*}
to \eqref{phi->F:matrix}, we have 
\begin{equation}
    \frac{\der}{\der s}
    \begin{pmatrix} F_0 \\ F_1 \\ F_2 \\ \vdots \end{pmatrix}
    =
    \left.\frac{T}{e^T-1}\right|_{T=\ders \Lambda^{-1}}
    \begin{pmatrix} \phi_0 \\ \phi_1 \\ \phi_2 \\ \vdots \end{pmatrix},
\label{dFn/ds:matrix}
\end{equation}
or equivalently, 
\begin{equation}
    \frac{\der F_n}{\der s}
    =
    \phi_n - \frac{\phi_{n-1}}{2}
    +
    \sum_{p=1}^{[n/2]} K_{2p} \phi_{n-2p},
\label{dFn/ds}
\end{equation}
where $K_{2p}$ is determined by \eqref{def:K2p}. The system of first
order differential equations \eqref{dFn/dtj}, \eqref{dFn/dtbarj} and
\eqref{dFn/ds} may be understood as defining equations of
$F_n(t,\tbar,s)$. This system is integrable and determines $F_n$ up to
integration constants, because the system \eqref{dtau/dtn},
\eqref{dtau/dtbarn} and \eqref{tau(s+hbar)} are compatible.

\begin{rem}
\label{rem:genus-exp}
Tau functions in string theory and random matrices 
are known to have a {\em genus expansion} of the form
\begin{equation}
    \log \tau = \sum_{g=0} \hbar^{2g-2} \calF_g,
\label{genus-exp}
\end{equation}
where $\calF_g$ is the contribution from Riemann surfaces of genus $g$.
In contrast, general tau functions of the $\hbar$-dependent Toda hierarchy
is not of this form, namely, odd powers of $\hbar$ can appear in the
$\hbar$-expansion of $\log\tau$. To exclude odd powers therein, we need
to impose conditions
\begin{equation*}
 \begin{split}
    0 &=
    v_{2m+1,j} 
    + \sum_{\substack{k+l=j\\ k\geq 1, l\geq 1}}
      \frac{1}{l}\, \frac{\der v_{2m,l}}{\der t_k}
    =
    \vbar_{2m+1,j} + \frac{\der\phi_{2m+1}}{\der \tbar_j}
    + \sum_{\substack{k+l=j\\ k\geq 1, l\geq 1}}
      \frac{1}{l}\, \frac{\der \vbar_{2m,l}}{\der \tbar_k}
\\
    &=
    \phi_{2m+1} - \frac{\phi_{2m}}{2}
    +
    \sum_{p=1}^{m} K_{2p} \phi_{2m+1-2p},
 \end{split}
\end{equation*}
 on $v_{n,j}$, $\vbar_{n,j}$ and $\phi_n$ or
\begin{equation}
 \begin{split}
    0 &=
    \frac{\der S_{2m+1}}{\der z}
    - 
    \sum_{j=1}^\infty z^{-j-1} \frac{\der S_{2m}}{\der t_j}
    =
    \frac{\der \Sbar_{2m+1}}{\der \zbar}
    +
    \sum_{j=1}^\infty \zbar^{j-1} \frac{\der \Sbar_{2m}}{\der \tbar_j}
\\
    &=
    \phi_{2m+1} - \frac{\phi_{2m}}{2}
    +
    \sum_{p=1}^{m} K_{2p} \phi_{2m+1-2p},
 \end{split}
\end{equation}
 on $S_n$, $\Sbar_n$ and $\phi_n$.
%
\end{rem}

\appendix

\section{Example ($c=1$ string theory)}
\label{sec:kontsevich-model}

In this appendix, we apply our algorithm to the compactified $c=1$
string theory at a self-dual radius, following the formulation in
\cite{tak:96}. We use the notations in \secref{sec:recursion}.

According to (4.10) in \cite{tak:96} ($\beta=1$), the
string equation for this case is
\begin{equation}
 \begin{split}
    L &= (-\Mbar - \hbar + 1)\Lbar,
\\
    \Lbar^{-1} &= (-M + 1) L^{-1}.
 \end{split}
\label{string-eq}
\end{equation}
Multiplying the left and right hand side of the second equation to the
right and left hand side of the first equation from the right, we have
\begin{equation*}
    L(-M+1)L^{-1} = -\Mbar-\hbar+1,
\end{equation*}
and using the canonical commutation relation $[L,M]=\hbar L$, we have
$M=\Mbar$, namely, 
\begin{equation}
    L = (1 - \Mbar - \hbar) \Lbar, \qquad
    M = \Mbar.
\label{string-eq:canonical}
\end{equation}
Hence the data $(f,g,\fbar,\gbar)$ for \propref{prop:RH} in this case are
\begin{equation}
 \begin{aligned}
    f(\hbar,s,e^{\hbar\ders}) &= e^{\hbar\ders}, &\qquad
    g(\hbar,s,e^{\hbar\ders}) &= s,
\\
    \fbar(\hbar,s,e^{\hbar\ders}) &= (1-s-\hbar)e^{\hbar\ders}, &
    \gbar(\hbar,s,e^{\hbar\ders}) &= s.
\end{aligned}
\label{f,g,fbar,gbar:string}
\end{equation}
The corresponding dispersionless data $(f_0,g_0,\fbar_0,\gbar_0)$ for
\propref{prop:dRH} are
\begin{equation}
 \begin{aligned}
    f_0(s,\xi) &= \xi, & \qquad
    g_0(s,\xi) &= s, 
\\
    \fbar_0(s,\xi) &= (1-s)\xi, &
    \gbar_0(s,\xi) &= s.
 \end{aligned}
\label{f0,g0,fbar0,gbar0:string}
\end{equation}
For the sake of simplicity, we fix the time variables
$\tbar_n$ ($n=1,2,\dotsc$) to $0$, which makes it possible to determine
all $X_n$'s explicitly, \eqref{Xn,Xbarn,phin:string}. If we turn on
$\tbar_n$'s, we need to proceed perturbatively.

\bigskip
To begin with, let us determine the leading terms of $X$, $\Xbar$ and
$\phi$ with respect to the $\hbar$-order, namely $X_0$, $\Xbar_0$ and
$\phi_0$ in \eqref{X:h-expansion}, \eqref{Xbar:h-expansion} and
\eqref{phi:h-expansion}. 

The Riemann-Hilbert type problem for $(\calL,\calM,\calLbar, \calMbar)$
\eqref{(f0,g0)(M,L)=(fbar0,gbar0)(Mbar,Lbar)} is
\begin{equation}
    \calL = (1-\calMbar) \calLbar, \qquad \calM=\calMbar.
\label{dRH:string}
\end{equation}
Recall that $\calL$, $\calM$, $\calLbar$ and $\calMbar$ have the
following form by \eqref{calL}, \eqref{calLbar}, \eqref{calM} and
\eqref{calMbar} when $\tbar=0$.
\begin{align*}
    \calL &=
    \xi
    + \sum_{n=0}^\infty u_{0,n+1} \xi^{-n},
\\
    \calLbar &=
    \sum_{n=0}^\infty \utilde_{0,n} \xi^{n+1},
\\
    \calM &= \sum_{n=1}^\infty nt_n \calL^n + s
    + \alpha_0
    + \sum_{n=1}^\infty v_{0,n} \calL^{-n},
\\ 
   \calMbar &= 
    s + \alphabar_0 
    + \sum_{n=1}^\infty \vbar_{0,n} \calLbar^n,
\end{align*}
Therefore $(1-\calMbar)\calLbar$ is a Taylor series with positive powers
of $\xi$, while $\calL$ does not have a positive power of $\xi$ except
for the first term, $\xi$. Therefore the first equation in
\eqref{dRH:string} implies that
\begin{equation}
    \calL = \xi.
\label{calL=xi:string}
\end{equation}
{}From this and the second equation $\calM=\calMbar$ in
\eqref{dRH:string}, it follows that $\calM$ and $\calMbar$ do not have
negative powers of $\xi$ and $\alpha_0=\alphabar_0$. Hence we may assume
that $\alpha_0=\alphabar_0=0$ and
\begin{equation}
    \calM = \calMbar = s + \sum_{n=1}^\infty n t_n \xi^n.
\label{calM=calMbar:string}
\end{equation}
Substituting this into the first equation of \eqref{dRH:string}, we have
\begin{equation}
    \calLbar = \xi \left(1-s-\sum_{n=1}^\infty n t_n \xi^n\right)^{-1},
    \text{\ or\ }
    \calLbar^{-1}
    = \xi^{-1} \left(1-s-\sum_{n=1}^\infty n t_n \xi^n\right).
\label{calLbar:string}
\end{equation}

Next, let us determine the leading terms $X_0$, $\Xbar_0$ and $\phi_0$
of the dressing operators $X$, $\Xbar$ and $\phi$. We denote the symbols
of $X_0$ and $\Xbar_0$ by $\calX_0=\calX_0(t,s,\xi)$ and
$\calXbar_0=\calXbar_0(t,s,\xi)$. 

Since $\calL = \exp(\ad_{\{,\}} \calX_0) \xi = \xi$, $\calX_0$ does not
depend on $s$. On the other hand, since $\calM =
\exp(\ad_{\{,\}}\calX_0)\left(s + \sum_{n} n t_n \xi^n\right) = s +
\sum_n n t_n \xi^n$, $\calX_0$ does not depend on $\xi$, either, which
means that $\calX_0=0$. 

Note that $\ad_{\{,\}} \phi_0(s)$ does not change the degree of
homogeneous terms with respect to $\xi$, since $\phi_0(s)$ does not
contain $\xi$. Hence $\calLbar$ has the following asymptotic behaviour
around $\xi=0$.
\begin{equation*}
    \calLbar = e^{\ad_{\{,\}}\phi_0} \xi 
    + e^{\ad_{\{,\}}\phi_0} \sum_{N=1}^\infty \frac{1}{N!}
    \bigl(\ad_{\{,\}}\calXbar_0\bigr)^N \xi
    =  e^{\ad_{\{,\}}\phi_0} \xi + O(\xi^2),
\end{equation*}
because $\calXbar_0$ is a Taylor series of $\xi$ with positive
powers. Comparing this expansion with \eqref{calLbar:string}, we have
\begin{equation}
    e^{\ad_{\{,\}}\phi_0} \xi = (1-s)^{-1} \xi.
\label{e(ad(phi0))xi}
\end{equation}
It is easy to see that the left hand side is equal to
$e^{-\phi_0'(s)}\xi$, where ${}'$ denotes the derivation by $s$. Thus we
obtain
\begin{equation}
    \phi_0(s) = \int^s \log(1-s) \, ds
    = -(1-s) \log(1-s) + (1-s).
\label{phi0:string}
\end{equation}
It remains to determine $\calXbar_0$. Operating $e^{-\ad_{\{,\}}\phi_0}$
to $\calLbar^{-1}$ \eqref{calLbar:string} and $\calMbar$
\eqref{calM=calMbar:string} and using the formula
\begin{equation}
    e^{-\ad_{\{,\}}\phi_0} \xi = (1-s) \xi,
\label{e(-ad(phi0))xi}
\end{equation}
which follows directly from \eqref{e(ad(phi0))xi}, we have two equations
characterising $\calXbar_0$:
\begin{equation}
 \begin{aligned}
    e^{\ad_{\{,\}}\calXbar_0} \xi^{-1}
    &=
    \xi^{-1} - \sum_{n=1}^\infty n t_n (1-s)^{n-1} \xi^{n-1},
\\
    e^{\ad_{\{,\}}\calXbar_0} s
    &=
    s + \sum_{n=1}^\infty n t_n (1-s)^n \xi^n.
 \end{aligned}
\label{e(ad(Xbar0))}
\end{equation}
In fact we can determine $\calXbar_0$ explicitly as follows.
\begin{equation}
    \calXbar_0 = \sum_{n=1}^\infty t_n (1-s)^n \xi^n.
\label{Xbar0:string}
\end{equation}
Indeed, since
\begin{equation*}
    \{\calXbar_0, \xi^{-1}\}
    =
    - \sum_{n=1}^\infty n t_n (1-s)^{n-1} \xi^{n-1}
    \text{\ and\ }
    \{\calXbar_0, s\}
    =
    \sum_{n=1}^\infty n t_n (1-s)^n \xi^n
\end{equation*}
commute with $\calXbar_0$ itself (this is a direct consequence of a
trivial fact $\{(1-s)^k \xi^k , (1-s)^l\xi^l\} = 0$), the exponentials
in $e^{\ad_{\{,\}}\calXbar_0} \xi^{-1}$ and
$e^{\ad_{\{,\}}\calXbar_0} s$ are truncated up to the first order, 
namely 
\begin{equation}
 \begin{aligned}
    e^{\ad_{\{,\}}\calXbar_0} \xi^{-1}
    &=
    \xi^{-1} + 
    \{\calXbar_0, \xi^{-1}\}
    =
    \xi^{-1} - \sum_{n=1}^\infty n t_n (1-s)^{n-1} \xi^{n-1},
\\
    e^{\ad_{\{,\}}\calXbar_0} s
    &=
    s + \{\calXbar_0, s\}
    =
    s + \sum_{n=1}^\infty n t_n (1-s)^n \xi^n,
 \end{aligned}
\end{equation}
which proves that $\calXbar_0$ in \eqref{Xbar0:string} satisfies
\eqref{e(ad(Xbar0))}. 

Thus we have determined the leading terms of $X$, $\Xbar$ and $\phi$ as
follows: 
\begin{equation}
    X_0 = 0, \quad
    \Xbar_0 = \sum_{n=1}^\infty t_n (1-s)^n e^{n\hbar\ders}, \quad
    \phi_0 = -(1-s) \log(1-s) + (1-s).
\label{X0,Xbar0,phi0:string}
\end{equation}

\bigskip
Having determined $X_0$, $\Xbar_0$ and $\phi_0$, we can start the algorithm
discussed in \secref{sec:recursion}. Following the procedure
by straightforward computation (actually, not so straightfoward, as we
shall see later), we obtain as the first and the second steps,
\begin{equation}
    X_1 = 0, \quad
    \Xbar_1 = - \sum_{n=1}^\infty 
              t_n \frac{n(n+1)}{2} (1-s)^{n-1} e^{n\hbar\ders}, \quad
    \phi_1 = \frac{1}{2} \log(1-s),
\label{X1,Xbar1,phi1:string}
\end{equation}
and 
\begin{equation}
    X_2 = 0, \quad
    \Xbar_2 = - \sum_{n=1}^\infty 
              t_n \frac{n(n^2-1)(3n+2)}{24} (1-s)^{n-2} e^{n\hbar\ders}, 
    \quad
    \phi_2 = -\frac{1}{12} (1-s)^{-1}.
\label{X2,Xbar2,phi2:string}
\end{equation}
{}From these results we can infer the Ansatz for all $n\geqq 2$:
\begin{equation}
    X_n = 0, \quad
    \Xbar_n = \sum_{m=1}^\infty 
              t_m c_{n,m} (1-s)^{m-n} e^{m\hbar\ders}, 
    \quad
    \phi_n = c_{n,0} (1-s)^{-n+1},
\label{Xn,Xbarn,phin:string}
\end{equation}
with suitable constants $c_{n,m}$ ($n,m\geqq 1$) and $c_{n,0}$ ($n\geqq
2$).  Eventually these constants are determined recursively as follows:
\begin{align}
    c_{n,m} &= \frac{1}{n}
    \sum_{j=0}^{n-1} (-1)^{n-j} \binom{m-j+1}{k-j+1} c_{j,m},
\label{cnm}
\\
    c_{n,0} &= \frac{1}{-n+1}
    \left(\frac{1}{n+1} - \frac{c_{1,0}}{n} -
    \sum_{j=2}^{n-1} (-1)^{n-j} \binom{-j+1}{k-j+1} c_{j,0}
    \right),
\label{cn0}
\end{align}
with the initial values $c_{0,m}=1$ and $c_{1,0}=1/2$.

Let us prove that the above Ansatz is true. To do this, we have only to
check that it is consistent with the algorithm in
\secref{sec:recursion}. 

It is easy to compute the intermediate objects $P^{(i-1)}$
\eqref{def:Pi-1} and $Q^{(i-1)}$ \eqref{def:Qi-1}, since the operators
$X_0,\dotsc,X_{i-1}$ are zero. Using the notations \eqref{ft,gt} and
\eqref{def:X,Xbar,phi(i-1)}, we have
\begin{equation}
 \begin{split}
    P^{(i-1)} &=
    \exp\left(\frac{X^{(i-1)}}{\hbar}\right) f_t
    = f_t
    = e^{\hbar\ders},
\\
    Q^{(i-1)} &=
    \exp\left(\frac{X^{(i-1)}}{\hbar}\right) g_t
    = g_t
    = s + \sum_{n=1}^{\infty} n t_n e^{n\hbar\ders}.
 \end{split}
\label{P,Q(i-1):string}
\end{equation}
Hence the terms in the expansion \eqref{Pi-1,Qi-1:h-expand} vanish
except for $P^{(i-1)}_0$ and $Q^{(i-1)}_0$:
\begin{align}
    P^{(i-1)}_0 &= e^{\hbar\ders}, &
    P^{(i-1)}_1 &= P^{(i-1)}_2 = \dotsb = P^{(i-1)}_i = 0,
\label{P(i-1)k:string}
\\
    Q^{(i-1)}_0 &= s + \sum_{n=1}^{\infty} n t_n e^{n\hbar\ders}, &
    Q^{(i-1)}_1 &= Q^{(i-1)}_2 = \dotsb = Q^{(i-1)}_i = 0.
\label{Q(i-1)k:string}
\end{align}
Their symbols are
\begin{align}
    \calP^{(i-1)}_0 &= \xi, &
    \calP^{(i-1)}_1 &= \calP^{(i-1)}_2 = \dotsb = \calP^{(i-1)}_i = 0,
\label{calP(i-1)k:string}
\\
    \calQ^{(i-1)}_0 &= s + \sum_{n=1}^{\infty} n t_n \xi^n, &
    \calQ^{(i-1)}_1 &= \calQ^{(i-1)}_2 = \dotsb = \calQ^{(i-1)}_i = 0.
\label{calQ(i-1)k:string}
\end{align}
To compute $\Pbar^{(i-1)}$ \eqref{def:Pbari-1}, let us consider
$\exp(\ad(\Xbar^{(i-1)}/\hbar))\fbar_\tbar$ first. Note that
\begin{equation}
 \begin{split}
    \hbar^{-1}[\Xbar^{(i-1)}, \fbar] &=
    \sum_{n=0}^{i-1} \hbar^{n-1} [\Xbar_n, \fbar]
\\
    &=
    \sum_{n=0}^{i-1} \hbar^{n-1}
    \sum_{m=1}^\infty t_m c_{n,m} 
    [(1-s)^{m-n}e^{m\hbar\ders}, (1-s-\hbar)e^{\hbar\ders}]
 \end{split}
\end{equation}
Substituting
\begin{multline}
    [(1-s)^{m-n}e^{m\hbar\ders}, (1-s-\hbar)e^{\hbar\ders}]
\\
    = \left(-n \hbar (1-s)^{m-n}
    - \sum_{r=2}^{m-n+1}\binom{m-n+1}{r} (-\hbar)^r (1-s)^{m-n+1-r}
    \right) e^{(m+1)\hbar\ders},
\end{multline}
we have
\begin{equation*}
 \begin{split}
    &\hbar^{-1}[\Xbar^{(i-1)}, \fbar]
\\
    =&{}
    - \sum_{n=0}^{i-1} \hbar^n
    \sum_{m=1}^\infty n t_m c_{n,m} (1-s)^{m-n} e^{(m+1)\hbar\ders}
\\
    &+ \sum_{k=0}^\infty \hbar^k
    \sum_{m=1}^\infty t_m
    \sum_{\substack{0\leq j \leq i-1 \\ r \geq 2 \\ j-1+r=k}}
    (-1)^{r+1} c_{j,m} \binom{m-j+1}{r}
    (1-s)^{m-k} e^{(m+1)\hbar\ders}.
 \end{split}
\end{equation*}
Because of the definition of $c_{n,m}$ \eqref{cnm}, the coefficients of
$\hbar^n t_m$ ($n=0,\dotsc,i-1$, $m=1,2,\dotsc$) in the right hand side
vanish and the coefficient of $\hbar^i$ is equal to $i
c_{i,m}$. (Actually this is why $c_{n,m}$ is defined by the recursion
relation \eqref{cnm}.) This means
\begin{equation}
    \hbar^{-1}[\Xbar^{(i-1)}, \fbar]
    = \hbar^i
    \sum_{m=1}^\infty t_m i c_{i,m} (1-s)^{m-i} e^{(m+1)\hbar\ders}
    + O(\hbar^{i+1}).
\label{[Xbar(i-1),fbar]}
\end{equation}
Further application of $\ad (\hbar^{-1} \Xbar^{(i-1)})$ changes the
symbol of terms of $\hbar$-order $-i$ (i.e., coefficients of $\hbar^i$)
by application of $\ad_{\{,\}} \calXbar_0$, as $\ad \hbar^{j-1}\Xbar_j$
($j=1,\dotsc,i-1$) lowers the $\hbar$-order. Hence for $N\geqq 1$ we
have 
\begin{multline}
    \bigl(\ad \hbar^{-1} \Xbar^{(i-1)}\bigr)^N \fbar
\\
    = \hbar^i \bigl(\ad_{\{,\}}\calXbar_0\bigr)^{N-1}\left.\left(
    \sum_{m=1}^\infty t_m i c_{i,m} (1-s)^{m-i} \xi^{m+1}
    \right)\right|_{\xi \to e^{\hbar\ders}}
    + O(\hbar^{i+1}).
\label{ad(Xbar(i-1))N(fbar)}
\end{multline}

Next we compute the conjugation of $\fbar$ by
$e^{\hbar^{-1}\phi^{(i-1)}(s)}$. 
\begin{equation}
 \begin{split}
    e^{\ad \hbar^{-1}\phi^{(i-1)}(s)} \fbar
    ={}&
    e^{\hbar^{-1}\phi^{(i-1)}(s)} (1-s-\hbar) e^{\hbar\ders}
    e^{-\hbar^{-1}\phi^{(i-1)}(s)}
\\
    ={}&
    e^{\hbar^{-1}(\phi^{(i-1)}(s) - \phi^{(i-1)}(s+\hbar))}
    (1-s-\hbar) e^{\hbar\ders}
\\
    ={}&
    \exp\Biggl(
    \frac{1}{\hbar}(\phi_0(s)-\phi_0(s+\hbar)) + \log(1-s-\hbar) 
\\
    &\qquad+
    \sum_{j=1}^{i-1} \hbar^{j-1}(\phi_j(s)-\phi_j(s+\hbar))
    \Biggr) e^{\hbar\ders}.
 \end{split}
\end{equation}
By \eqref{X0,Xbar0,phi0:string}, \eqref{X1,Xbar1,phi1:string} and
\eqref{Xn,Xbarn,phin:string}, we have
\begin{align*}
    \frac{1}{\hbar}(\phi_0(s)-\phi_0(s+\hbar)) &+ \log(1-s-\hbar)
    =
    - \sum_{k=1}^\infty \frac{\hbar^k}{k+1} (1-s)^{-k},
\\
    \phi_1(s) - \phi_1(s+\hbar) &=
    \sum_{k=1}^\infty \hbar^k \frac{c_{1,0}}{k} (1-s)^{-k},
\\
    \phi_j(s) - \phi_j(s+\hbar) &=
    \sum_{k=1}^\infty \hbar^k 
    (-1)^{k+1} c_{j,0}\binom{-j+1}{k} (1-s)^{-j-k+1}.
\end{align*}
The coefficients $c_{j,0}$ were defined by \eqref{cn0} so that 
\begin{multline*}
    \frac{1}{\hbar}(\phi_0(s)-\phi_0(s+\hbar)) + \log(1-s-\hbar) 
    + \sum_{j=1}^{i-1} \hbar^{j-1}(\phi_j(s)-\phi_j(s+\hbar))
\\
    = (i-1) c_{i,0} \frac{\hbar^i}{(1-s)^i}
    + O(\hbar^{i+1}).
\end{multline*}
Thus we have
\begin{equation}
    e^{\ad \hbar^{-1}\phi^{(i-1)}(s)} \fbar
    =
    \bigl(1 + \hbar^i (i-1)c_{i,0} (1-s)^{-i} + O(\hbar^{i+1})\bigr)
    e^{\hbar\ders},
\label{exp(ad(phi(i-1)))fbar}
\end{equation}
for $i\geqq 2$. When $i=1$, we should replace $(1-s)^{-1}$ by
$\log(1-s)$ but the rest is the same.

Summarising the above results \eqref{ad(Xbar(i-1))N(fbar)} and
\eqref{exp(ad(phi(i-1)))fbar}, we obtain the following expansion of
$\Pbar^{(i-1)}$.
\begin{equation}
 \begin{split}
    &\Pbar^{(i-1)}
\\  ={}&
    e^{\ad \hbar^{-1}\phi^{(i-1)}} e^{\ad \hbar^{-1}\Xbar^{(i-1)}} \fbar
\\
    ={}&
    e^{\ad \hbar^{-1}\phi^{(i-1)}} \fbar
    +
    e^{\ad \hbar^{-1}\phi^{(i-1)}}
    \sum_{N=1}^\infty 
    \frac{\bigl(\ad \hbar^{-1}\Xbar^{(i-1)}\bigr)^N}{N!} \fbar
\\
    ={}&
    e^{\hbar\ders} + \hbar^i (i-1) c_{i,0} (1-s)^{-i} e^{\hbar\ders}
\\
    &+ \hbar^i
    e^{\ad_{\{,\}}\phi_0}
    \sum_{N=1}^\infty 
    \frac{\bigl(\ad_{\{,\}}\calXbar_0\bigr)^{N-1}}{N!}
    \left.\left(
     \sum_{m=1}^\infty t_m i c_{i,m} (1-s)^{m-i} \xi^{m+1}
    \right)\right|_{\xi \to e^{\hbar\ders}}
\\
    &+ O(\hbar^{i+1}).
 \end{split}
\label{Pbar(i-1):string}
\end{equation}
Therefore 
\begin{equation}
    \Pbar^{(i-1)}_0 = e^{\hbar\ders}, \qquad
    \Pbar^{(i-1)}_1 = \dotsb = \Pbar^{(i-1)}_{i-1} = 0,
\label{Pbar(i-1)k:string}
\end{equation}
which coincide with \eqref{P(i-1)k:string}, and
\begin{multline}
    \Pbar^{(i-1)}_i
    = (i-1) c_{i,0} (1-s)^{-i}\xi
\\    + 
    e^{\ad_{\{,\}}\phi_0}
    \sum_{N=1}^\infty 
    \frac{\bigl(\ad_{\{,\}}\calXbar_0\bigr)^{N-1}}{N!}
    \left(
     \sum_{m=1}^\infty t_m i c_{i,m} (1-s)^{m-i} \xi^{m+1}
    \right).
\label{Pbar(i-1)i:string:temp}
\end{multline}
{}From formulae
\begin{equation*}
    \ad_{\{,\}}\calXbar_0 \bigl((1-s)^{-j}\bigr)
    =
    j (1-s)^{-j-1} 
    \left(\sum_{k=1}^\infty k t_k (1-s)^k \xi^k\right),
\end{equation*}
$\{\calXbar_0, (1-s)^k\xi^k\}=0$ and \eqref{e(ad(phi0))xi} it follows
that
\begin{multline}
    e^{\ad_{\{,\}}\phi_0} \bigl(\ad_{\{,\}}\calXbar_0\bigr)^{N-1} 
    \bigl((1-s)^{m-i}\xi^{m+1}\bigr)
\\
    =
    (i+1) \dotsm (i+N-1) (1-s)^{-i-N} \xi^{m+1}
    \left(\sum_{k=1}^\infty k t_k \xi^k\right)^{N-1}.
\label{e(ad(phi0))(ad(Xbar0))N-1(1-s)xi}
\end{multline}
Substituting it to \eqref{Pbar(i-1)i:string:temp}, we have
\begin{multline}
    \calPbar^{(i-1)}_i
    = (i-1) c_{i,0} (1-s)^{-i} \xi
\\
    + \sum_{N=1}^\infty \frac{1}{N!}
    \sum_{m=1}^\infty t_m c_{i,m} \xi^{m+1} i (i+1) \dotsm (i+N-1) \times
\\
    \times
    (1-s)^{-i-N}
    \left(\sum_{k=1}^\infty k t_k \xi^k\right)^{N-1}.
\label{Pbar(i-1)i:string}
\end{multline}

Now let us compute \eqref{tildeXi=ints}. Although we have not computed
$\Qbar^{(i-1)}$, thanks to \eqref{calP(i-1)k:string} and
\eqref{calQ(i-1)k:string}, integrals in \eqref{tildeXi=ints} are
simplified to
\begin{equation*}
    -\tilde\calX_i+\phi_i+\tilde\calXbar_i
    =
    \int^s \xi^{-1} \calPbar^{(i-1)}_i \, ds,
\end{equation*}
which is computable without information of $\Qbar^{(i-1)}$. By the
explicit formula \eqref{Pbar(i-1)i:string} we obtain
\begin{multline}
    -\tilde\calX_i+\phi_i+\tilde\calXbar_i
    =
    c_{i,0} (1-s)^{-i+1}
\\
    +
    \sum_{N=1}^\infty \frac{1}{N!}\sum_{m=1}^\infty
    t_m c_{i,m} \xi^m i (i+1) \dotsm (i+N-2) (1-s)^{-i-N+1}
    \left(\sum_{k=1}^\infty k t_k \xi^k\right)^{N-1}.
\end{multline}
In this formula there is no term with negative powers of $\xi$, which
means $\tilde\calX_i=0$, i.e., $\calX_i=0$. The constant term with
respect to $\xi$ is $c_{i,0} (1-s)^{-i+1}$, which is $\phi_i(s)$, as was
expected. The remaining part is $\tilde\calXbar_i$. 

In general, it is almost hopeless to compute $\calXbar_i$ from
$\tilde\calXbar_i$ by \eqref{tildeXi,Xbari->Xi,Xbari:symbol}. However,
quite fortunately, in the present case we are able to find the explicit
answer. Using \eqref{e(ad(phi0))(ad(Xbar0))N-1(1-s)xi}, we can rewrite
$\tilde\calXbar_i$ as follows.
\begin{equation}
    \tilde\calXbar_i
    =
    e^{\ad_{\{,\}}\phi_0}\left(
    \sum_{N=1}^\infty \frac{\bigl(\ad_{\{,\}}\calXbar_0\bigr)^{N-1}}{N!}
    \left(\sum_{m=1}^\infty t_m c_{i,m} (1-s)^{m-i} \xi^m\right)\right).
\label{tildeXbari':string}
\end{equation}
Recall that equations in \eqref{tildeXi,Xbari->Xi,Xbari:symbol} are the
inversion formulae of \eqref{def:tildeX} and
\eqref{def:tildeXbar}. Comparing \eqref{tildeXbari':string} and
\eqref{def:tildeXbar}, we can conclude that
\begin{equation*}
    \calXbar_i =
    \sum_{m=1}^\infty t_m c_{i,m} (1-s)^{m-i} \xi^m,
\end{equation*}
which finally proves the Ansatz \eqref{Xn,Xbarn,phin:string}.


\begin{thebibliography}{dFGZ}

\bibitem[A]{aok:86}
           Aoki, T., Calcul exponentiel des op\'erateurs
	   microdiff\'erentiels d'ordre infini. II.
           Ann. Inst. Fourier (Grenoble)  {\bf 36}
	   143--165 (1986).

\bibitem[B]{bourbaki}
N. Bourbaki, \'El\'ements de math\'ematique,
Groupes et alg\`ebres de Lie, Actualit\'es Sci. Indust.,
(Hermann).

\bibitem[D]{dij:91}
Dijkgraaf, R.,
Intersection Theory, Integrable Hierarchies and 
Topological Field Theory,
in {\it New symmetry principles in quantum field theory} (Carg\`ese,
1991), 
NATO Adv. Sci. Inst. Ser. B Phys. {\bf 295} (1992), 95--158. 

\bibitem[dFGZ]{dfgz}
Di Francesco, P., Ginsparg, P., and Zinn-Justin, J.,
$2$D gravity and random matrices.
Phys. Rep. {\bf 254} (1995), 133 pp. 

\bibitem[DMP]{dmp:93}
R.~Dijkgraaf, G.~Moore and R.~Plesser, 
The partition function of 2D string theory, 
Nucl. Phys. {\bf B394} (1993), 356--382.  

\bibitem[EK]{egu-kan:94}
T.~Eguchi and H.~Kanno, 
Toda lattice hierarchy and the topological description 
of $c = 1$ string theory, 
Phys. Lett. {\bf B331} (1994), 330--334.  

\bibitem[EO]{eyn-ora:07}
Eynard, B., and Orantin, N.,
Invariants of algebraic curves and topological expansion,
Commun. Number Theory Phys. {\bf 1} 
(2007),
347--452. 

\bibitem[HOP]{hop:94}
A.~Hanany, Y.~Oz and R.~Plesser, 
Topological Landau-Ginzburg formulation 
and integrable structure of 2d string theory, 
Nucl. Phys. {\bf B425} (1994), 150--172. 

\bibitem[Kr]{kri:91}
Krichever, I.M.,
The dispersionless Lax equations and topological minimal models,
Commun. Math. Phys. {\bf 143} (1991), 415--426.

\bibitem[Mo]{mor:94}
Morozov, A.,
Integrability and Matrix Models
Phys. Usp. {\bf 37} (1994), 1--55.

\bibitem[OS]{orl-sch}
Orlov, A. Yu. and Schulman, E. I., Additional symmetries for integrable
equations and conformal algebra representation,
Lett. Math. Phys. {\bf 12} (1986), 171--179;
Orlov, A. Yu.,
Vertex operators, $\bar{\partial}$-problems, symmetries, variational
indentities and Hamiltonian formalism for $2+1$ integrable systems,
in: {\it Plasma Theory and Nonlinear and
Turbulent Processes in Physics\/}
(World Scientific, Singapore, 1988);
Grinevich, P. G., and Orlov, A. Yu.,
Virasoro action on Riemann surfaces, Grassmannians,
$\det\bar{\partial}_j$ and Segal Wilson $\tau$ function,
in: {\it Problems of Modern Quantum Field Theory\/}
(Springer-Verlag, 1989).


\bibitem[S]{sch:85}
Schapira, P.,
Microdifferential systems in the complex domain,
Grundlehren der mathematischen Wissenschaften {\bf 269},
Springer-Verlag,
Berlin-New York,
(1985)

\bibitem[T1]{tak:96}
Takasaki, K.,
Toda Lattice Hierarchy and Generalized String Equations.
Commun. Math. Phys. {\bf 181} (1996), 131--156.


\bibitem[TT1]{tak-tak:91}
Takasaki, K., and Takebe, T.,
SDiff(2) Toda equation---hierarchy, tau function, and symmetries.
Lett. Math. Phys. {\bf 23} (1991), 205--214. 

\bibitem[TT2]{tak-tak:95}
Takasaki, K., and Takebe, T.,
Integrable Hierarchies and Dispersionless Limit,
Rev. Math. Phys. {\bf 7} (1995), 743-803.

\bibitem[TT3]{tak-tak:09}
Takasaki, K., and Takebe, T.,
An $\hbar$-expansion of the KP hierarchy
   --- recursive construction of solutions,
{\tt arXiv:0912.4867}

\bibitem[TT4]{tak-tak:11}
Takasaki, K., and Takebe, T.,
$\hbar$-Dependent KP hierarchy,
to appear in Theoret. and Math. Phys., the Proceedings of the
``International Workshop on Classical and Quantum Integrable
Systems 2011'' (January 24-27, 2011 Protvino, Russia), 
{\tt arXiv:1105.0794}

\bibitem[UT]{uen-tak:84}
Ueno, K., and Takasaki, K.,
Toda lattice hierarchy,
in {\it Group Representations and Systems of Differential Equations},
K. Okamoto ed., Advanced Studies in Pure Math\. {\bf 4}
(North-Holland/Kinokuniya 1984), 1--95.



\end{thebibliography}
\end{document}